\documentclass[10pt,twocolumn,twoside]{IEEEtran}

\usepackage{todonotes}
\usepackage[ruled]{algorithm2e}

\usepackage{amsmath}
\usepackage{bm}
\usepackage{amssymb}
\usepackage{graphicx}
\usepackage{float}
\usepackage{color}
\usepackage{amsthm}
\usepackage{mathrsfs} 
\usepackage{mathtools}
\usepackage{cleveref}
\usepackage{interval}

\usepackage{subfig}
\usepackage[shortlabels]{enumitem}
\usepackage[switch]{lineno}

\DeclareMathOperator*{\argmin}{arg\,min}

\newcommand{\real}{\mathbb{R}}

\newcommand{\cfont}[1]{\texttt{#1}}

\newcommand{\lree}{\bar{J}_l}

\newcommand{\argse}[1]{ {x}^i_{#1}, {u}^i_{#1}}

\newcommand{\upstream}[1]{  {{#1}^i_{u}}}
\newcommand{\downstream}[1]{ {#1}^i_{d}}

\newcommand{\mc}[1]{\Theta_{m \ell_1}{#1}}
\newcommand{\uncon}[1]{\Theta_u{#1}}
\newcommand{\m}[1]{\Theta_m{#1}}
\newcommand{\elltwo}[1]{\Theta_{m \ell_2}{#1}}
\newcommand{\mMatrix}{\mathbb{M}}
\newcommand{\R}{\mathbb{R}}
\newcommand{\Nodes}{\mathscr{V}}
\newcommand{\Edges}{\mathscr{E}}

\newcommand{\tb}[1]{#1}

\usepackage{todonotes}
\usepackage{algorithmic}
\usepackage{amsmath}
\usepackage{bm}
\usepackage{amssymb}
\usepackage{graphicx}
\usepackage{float}
\usepackage{color}
\usepackage{amsthm}
\usepackage{mathrsfs} 
\usepackage{mathtools}
\usepackage{cleveref}

\usepackage{subfig}
\newcommand{\parDiff}[2]{\frac{\partial #1}{\partial #2}}

\newtheorem{theorem}{Theorem}

\newtheorem{definition}{Definition}
\newtheorem{lemma}{Lemma}

\newtheorem{remark}{Remark}
\newtheorem{proposition}{Proposition}

\title{Distributed Identification of Contracting and/or Monotone Network Dynamics}
\author{Max Revay, Jack Umenberger, Ian R. Manchester}

\begin{document}
	
	\maketitle
	
	\begin{abstract}
			This paper proposes methods for identification of large-scale networked systems with guarantees that the resulting model will be contracting -- a strong form of nonlinear stability -- and/or monotone, i.e. order relations between states are preserved. The main challenges that we address are: simultaneously searching for model parameters and a certificate of stability, and scalability to networks with hundreds or thousands of nodes. We propose a model set that admits convex constraints for stability and monotonicity, and has a separable structure that allows distributed identification via the alternating directions method of multipliers (ADMM). The performance and scalability of the approach is illustrated on a variety of linear and non-linear case studies, including a nonlinear traffic network with a 200-dimensional state space.
			
			
		
	\end{abstract}

	\section{ Introduction }

	
	System identification is the process of generating dynamic models from data \cite{ljung1999system}, and is also referred to as \textit{learning dynamical systems} (e.g. \cite{svensson2017flexible}).
	When scaling control and identification algorithms to large-scale systems, it can be useful to treat a system as a sparse network of local subsystems interconnected through a graph \cite{benner2004solving,siljak2011decentralized,van_den_hof_identification_2013}.
	In this paper, we propose algorithms for identification of such networked systems in state space form: 
	\begin{gather}
	x_{t+1} = a(x_t, u_t, u_{t+1}), \label{eq:global system}
	\end{gather}
	where $x_t \in \mathbb{R}^n$ and $u_t\in \mathbb{R}^m$  are the state and input respectively, and the model dynamics $a(\cdot,\cdot,\cdot)$ can be either linear or non-linear. 
	We assume that measurements (or estimates) of state and input sequences are available.
	
	Our approach:
	\begin{enumerate}
		\item uses distributed computation (i.e. network nodes only share data and parameters with immediate neighbors),
		\item can generate models with a strong form of stability called contraction,
		\item can generate monotone models, i.e. ordering relations between states are preserved.
	\end{enumerate}

\tb{	Imposing contraction and/or monotonicity on models provides two benefits when identifying systems that are known to satisfy those properties. Firstly, incorporating prior knowledge can significantly improve the quality of the identified models.
	Secondly, it guarantees that properties of the real system that are useful for controller design are present in the identified model.
	
	This work is motivated by the observation that many systems have the combination of large-scale, sparse dynamics, monotonicity and stability. Examples include traffic networks \cite{como2015throughput,lovisari2014stability,coogan2014dynamical}, chemical reactions \cite{de2007monotone}, combination therapies \cite{rantzer2014control,hernandez2011discrete,hernandez2013optimal,jonsson2014scalable}, wildfires \cite{somers2019priority} and power scheduling \cite{rantzer2014control}.
}

		
	
	The key technical difficulty we address is the simultaneous identification and stability verification of large-scale networked systems.
	We propose a convex model set with scalable stability conditions and an algorithm based on ADMM that decomposes the identification problem into easily solvable, sub-problems that require only local communication between subsystems.
	
	\subsection{Networked System Identification}
	Standard approaches to system identification do not work well for large-scale networked systems for three reasons \cite{haber2014subspace}:
	firstly, the dataset must be collected at a central location, a process which may be prohibitive for complex systems;
	secondly, the computational and memory complexities prohibit application to large systems; finally, the network structure may not be preserved by identification. For instance, standard subspace identification methods have $\mathcal{O}[n^3]$ an $\mathcal{O}[n^2]$ computational and memory complexities respectively, and any sparse structure in the dynamics is destroyed through an unknown similarity transformation \cite{verhaegen2007filtering}.

Previous work in networked system identification can be loosely categorized into two areas; the identification of a network topology \cite{materassi2010topological, materassi2012problem,sanandaji2011exact}, and the identification of a system's dynamics with known topology.
In the latter category, almost all prior work has focused on the case where subsystems are linear time invariant (LTI)  and described by state space models \cite{haber2014subspace,yu_subspace_2017,yu2019subspace} or transfer functions (a.k.a. modules) \cite{van_den_hof_identification_2018,dankers2016identification}.	

When identifying the subsytem dynamics, states or outputs of neighbors are treated as exogenous inputs, ignoring feedback loops induced by the network topology. This improves scalability as the identification of each subsystem can be performed in parallel. However, accurate identification of the individual subsystems does not imply accurate identification of the full network, because the ignored feedback loops may have a strong effect and even introduce instability. A simple case with two subsystems which has received significant attention is closed-loop identification \cite{forssell1999closed}



Prior works in networked system identification assume stable LTI network dynamics and establish identifiability \cite{weerts2018identifiability} and consistency \cite{van2013identification}. These assumptions then imply model stability in the infinite data limit. However, model stability is not guaranteed with finite data sets or in non-linear black-box identification problems, where the true system is usually not in the model set. 

	\subsection{Identification of Stable Models}	\label{intro:stable_id}  
	
	
Standard methods for system identification do not guarantee model stability, even if the system from which the data are collected is stable. For linear system identification considerable attention has been paid to this problem, and several methods have been suggested based on regularisation or model constraints \cite{van2001identification, maciejowski1995guaranteed, lacy2003subspace,miller2013subspace}. Even for linear systems, the set of stable models is not convex using standard parameterisations, to the authors' knowledge all existing methods introduce some bias in the identification procedure.

For linear systems most definitions of stability are equivalent. The nonlinear case is more nuanced, and the definition used depends on the requirements of the problem at hand. Standard Lyapunov methods are not appropriate in system identification as the stability certificate must be constructed about a known stable solution, whereas the very purpose of system identification is to predict a system's response to previously unseen inputs. Contraction \cite{lohmiller1998contraction} and incremental stability (e.g. \cite{tran2018convergence}) are more appropriate since they ensure stability of all possible solutions and consequently, do not require a-priori knowledge of the inputs and state trajectories.
		Stability guarantees have also been investigated for nonlinear system identification.
		For instance, systems can be identified using sets of stable recurrent neural networks \cite{miller_stable_2018} or stable Gaussian process state space models \cite{umlauft2017learning}. A limitation of these approaches is that they do not allow joint search for a model and its stability certificate, which can be conservative even for linear systems.

This paper builds on previous work in jointly-convex parameterization of models and their stability certificates via implicit models \cite{megretski2008convex,bond2010compact,tobenkin2010convex, tobenkin2017convex,revay2019contracting,revay2020convex} and associated convex bounds for model fidelity via Lagrangian relaxation \cite{tobenkin2017convex, umenberger2018specialized, umenberger2019convex}. 
\tb{The main development in this paper is to significantly improve scalability of this approach via a novel model parametrization and contraction constraint that are jointly convex and permit a particular upstream/downstream network decomposition (defined below).}

	\subsection{Monotone and Positive Systems} \label{intro:monotone_systems}
	
	Monotone systems are a class of dynamic system characterized by the preservation of an order relation for solutions (c.f. Definition \ref{def:monotone} below). A closely related class is \textit{positive systems}, for which state variables remain non-negative  for all non-negative inputs (c.f. Definition \ref{def:positive} below). For linear systems, positivity and monotonicity are equivalent.
	

	
A useful property of monotone systems is that they often admit simplified stability tests.
	In particular, for linear positive systems the existence of \textit{separable} Lyapunov functions, i.e. those representable as the sum or maximum over functions of individual state variables, is necessary and sufficient for stability \cite{berman1994nonnegative}. This property has been used to simplify analysis \cite{haddad2010nonnegative}, control \cite{rantzer2015scalable} and identification \cite{umenberger2016scalable} of positive systems. 
	Separable stability certificates have also been shown to exist for certain classes of nonlinear monotone systems \cite{hirsch2006monotone,dirr2015separable,manchester2017existence, kawano2019contraction}. and have been used for distributed stability verification \cite{coogan2019contractive} and control \cite{shiromoto2018distributed}. 
		Monotonicity can also simplify nonlinear model predictive control \cite{rantzer2014control} and formal verification using signal temporal logic \cite{sadraddini2018formal}.

	There are however, few identification algorithms that guarantee monotonicity. In \cite{shen2010estimation}, monotone gene networks are identified using the monotone P-splines developed in \cite{shen2011estimation}. This approach, however, does no guarantee model stability. 	

	\subsection{Least-Squares Equation Error}
	Identification typically involves the optimization of a quality of fit metric over a model set. In this paper we use what is arguably the simplest and most widely-applied quality-of-fit metric, least-squares equation error (a.k.a. one step ahead prediction error):
	\begin{equation} \label{eq:explicit_equation_error}
	J_{ee}(\theta) = \sum_{t=0}^{T-1} |a(\tilde{x}_{t}, \tilde{u}_{t}) - \tilde{x}_{t+1}|^2,
	\end{equation}
	where $\tilde{x}_t \in \mathbb{R}^n$ and $\tilde{u}_t\in \mathbb{R}^m$ are state and input measurements or estimates. 
	Least-squares equation error is a natural choice for short-term prediction if state measurements are available.
	
	If long-term predictions are needed, then simulation error, defined as
	\begin{gather}
	J_{se}(\theta) = \sum_{t=0}^{T-1} |x_t - \tilde{x}_t |^2 ,~ s.t. ~ x_{t+1} = a(x_t, \tilde{u}_t),
	\end{gather}
	is \tb{a better measure of performance}. The dependence on simulated states, however, renders the cost function non-convex \cite{sjoberg1995nonlinear,ljung1986system} \tb{and notoriously difficult to optimize \cite{ribeiro2020smoothness}. Consequently, } equation error optimization is often used to initialize local search methods (e.g. gradient descent) for models with good simulation error or used as a surrogate for simulation error \tb{with better numerical properties}.
	 In the latter context, model stability is particularly important since a model can have small equation error but be unstable and therefore exhibit very large simulation error. \tb{In fact, when a model is contracting, it can be shown that small equation error implies small simulation error \cite{megretski2008convex}.}
	
	In many contexts, system state measurements are not available. Nevertheless, equation error frequently arises as a sub-problem via estimated states, e.g. in subspace identification algorithms \cite{van1994n4sid,van2012subspace,yu2019subspace}, where states are estimated using using matrix factorizations, or in maximum likelihood identification via the expectation maximization (EM) algorithm where they are estimated from the joint smoothing distribution \cite{schon2011system, umenberger2018maximum}.

	\subsection{Contributions}
	The main contributions of this work as are follows: we propose a model structure and convex constraints that guarantee monotonicity, positivity, and/or contraction of the model. For large scale networked systems, we refine the model and constraints to have a separable structure, and we introduce a separable bound on equation error, so the identification problem can be solved using distributed computation.
The algorithm, based on ADMM,  decomposes into easily solved separable optimization problems at each step. Data and parameters are only communicated to immediate neighbours in the network. Finally, we evaluate the scalability and fitting performance of the method on a number of numerical examples.

	\section{Preliminaries and Problem Setup}


	\subsection*{Notation}	
	A graph $\mathscr{G}$ is defined by a set of nodes (vertices) $\mathscr{V} = [1,...,N]$ and edges $\mathscr{E} \subset \mathscr{V\times V}$.
	The vector $\bm{1}$ is the column vector of ones, with size inferred from context.
	For vectors $v$, $v> 0 $ refers to the element-wise inequality. For matrices $M$, $M\geq0$ and $M\leq0$ refer to element-wise inequalities.
	For symmetric matrices $M$, $M\succ 0 $ means that $M$ is positive definite.
	For a vector $v$, $diag(v)$ is the matrix with the elements of $v$ along the diagonal.
	The set of $n \times n$ symmetric matrices is denoted $\mathbb{S}^{n\times n}$.
	The set of $n \times n$ non-singular M-matrices is denoted $\mMatrix^n$. For a matrix $A$, $A\in \mMatrix^n$ means $A^{ij} \leq 0, ~\forall i \neq j$ and $\text{real}(\lambda_i) > 0$ for $i = 1,...,n$, where $\lambda_i$ are the eigenvalues of $A$. 	
	For brevity, we will sometimes drop the arguments from a function where the  meaning may be inferred from context.

	\subsection{Differential Dynamics}
	
	The contraction and monotonicity conditions we study can be verified by way of a systems \textit{differential dynamics}, a.k.a. linearized, variational, or prolonged dynamics. For the system \eqref{eq:global system}, the differential dynamics are
	\begin{equation}\label{eq:explicit_diff_dynamics}
	\delta_{x_{t+1}} = A(x_t, u_t, u_{t+1}) \delta_{x_{t}} + B(x_t, u_t, u_{t+1}) \delta_{u_t}.
	\end{equation}
	where $A = \parDiff{a}{x}$ and $B=\parDiff{a}{u}$. In conjunction with \eqref{eq:global system}, the differential dynamics describe the linearized dynamics along \textit{all} solutions of the system.

	\subsection{Contraction Analysis} 	\label{sec:intro contraction}
	
	We use the following definition of nonlinear stability:
	\begin{definition}[Contraction] \label{def:Incremental Lp Stability}
		A system is termed contracting with rate $\alpha$, where $0<\alpha < 1$, if for any two initial  conditions $x^a_0$, $x^b_0$,  given the same input sequence $u_t$, and some $p \in [1,\infty]$,
		there exists a continuous function $b_p(x^a_0, x^b_0) > 0$ such that the corresponding trajectories $x^a_t, x^b_t$ satisfy \tb{$|x^a_t - x^b_t|_p < \alpha^t b_p(x^a_0, x^b_0)$}.
	\end{definition} \noindent 
	Contraction can be proven by finding a contraction metric which verifies conditions on the differential dynamics \cite{lohmiller1998contraction}. 
	A contraction metric is a function $V(t, x, \delta_x)$ such that:
	\begin{gather} 
	V(t, x, 0) = 0, ~ ~ V(t, x, \delta) \ge  \mu |\delta|_p, \\
	V(t+1, x_{t+1}, \delta_{x_{t+1}}) \leq \alpha V(t, x, \delta_x) .\label{eq:contraction_metric}
	\end{gather}
	for some $\mu>0$
		
	The choice of contraction metric $V(t,x,\delta)$ is problem dependent. Prior works have proposed quadratic contraction metrics for which \eqref{eq:contraction_metric} is linear in the stability certificate and can be verified using semi-definite programming. A number of works have also noted that using a weighted $\ell_1$ norm can lead to separable constraints \cite{russo2010stability,coogan2019contractive} allowing for stability verification of large-scale networked systems.

In the context of system identification, the joint search for model $a$ in \eqref{eq:global system} and contraction metric $V$ is non-convex due to the  nonlinear function \tb{composition $V(t+1, x_{t+1}, \delta_{x_{t+1}}) =V(t+1, a(x,u), A(x,u)\delta_{x_t})$}.
	
	
	
	\subsection{Monotone and Positive Systems}
We now define system monotonicity and positivity of dynamical systems.
	\begin{definition}[Monotone System] \label{def:monotone}
		A system \eqref{eq:global system} is termed monotone if for inputs $u^a_t$ and $u^b_t$ and initial conditions $x^a_0$, $x^b_0$, the following implication holds:
		$$ x^a_0 \geq x^b_0, ~u^a_t \geq u^b_t~ \forall t  \implies x^a_t \geq x^b_t~ \forall t.$$
	\end{definition} 	\noindent 
	Monotonicity results from $A(x,u) \ge 0$  and $B(x,u) \ge 0$ where $A$ and $B$ come from the differential dynamics \eqref{eq:explicit_diff_dynamics}. 
	
	\begin{definition}[Positive System] \label{def:positive}
		A system \eqref{eq:global system} is positive if for all inputs $u_0,...,u_T \ge 0$ and initial conditions $x_0 \ge 0$, the resulting trajectory has $x_1,...,x_T \ge0$.
	\end{definition}
	A sufficent condition for a system to positive is for it to be monotone and admit $x_t=0, u_t=0 \,\forall t$ as a solution, i.e.  
	$a(0, 0, 0) = 0$ in \eqref{eq:global system}.
	
	

	\subsection{Network Structure}
	We assume model \eqref{eq:global system} is partitioned into $N$ subsystems. The interactions between these subsystems is described by a directed graph $\mathscr{G} = (\mathscr{V}, \mathscr{E})$. Here, we have a set of nodes denoted $\mathscr{V} = \{1,...,N\}$ corresponding to the subsystems.  Each subsystem has its own state denoted $x^i \in \mathbb{R}^{n_i}$ and may take an input denoted $u^i \in \mathbb{R}^{m_i}$ (we allow for the case $m_i = 0$). The global state and input is attained by concatenating the states and inputs of each subsystem,
	
	\begin{equation}\label{eq:global state}
	x = \begin{bmatrix}
	x^1\\ \vdots \\ x^N
	\end{bmatrix}, ~ ~
	u = \begin{bmatrix}
	u^1\\ \vdots \\ u^N
	\end{bmatrix}.
	\end{equation}

	The set of edges $\Edges \subseteq \Nodes\times \Nodes$ describes how the subsystems interact with each other. In particular, $(j,i) \in \Edges$ means that the state of subsystem $j$ affects the state of subsystem $i$. The edge list $\mathscr{E}$ may arise naturally from the context of the problem, e.g. in traffic networks where edges come from the physical topology of the road network, or may be identified from data \cite{materassi2010topological,materassi2013model}.
	
	For each subsystem $i \in \mathscr{V}$, we define the set of upstream neighbours $\upstream{\Nodes} = \{ j | (j,i) \in \Edges \}$ and the set of downstream neighbours $\downstream{\Nodes} = \{j | (i,j) \in \Edges\}$. The term \textit{upstream neighbours} of $i$ refers to the subsystems whose state affects the state of subsystem $i$, and the term \textit{downstream neighbours} refers to the subsystems whose state is  affected by subsystem $i$'s state. In general, we allow self-loops so that a node can be both upstream and downstream to itself. This notation is illustrated in Fig.  \ref{fig:Notation Illustration}.
	
	\begin{figure}
		\centering
		\includegraphics[trim = {0cm 10cm 0cm 3cm}, clip, width = 0.7\linewidth]{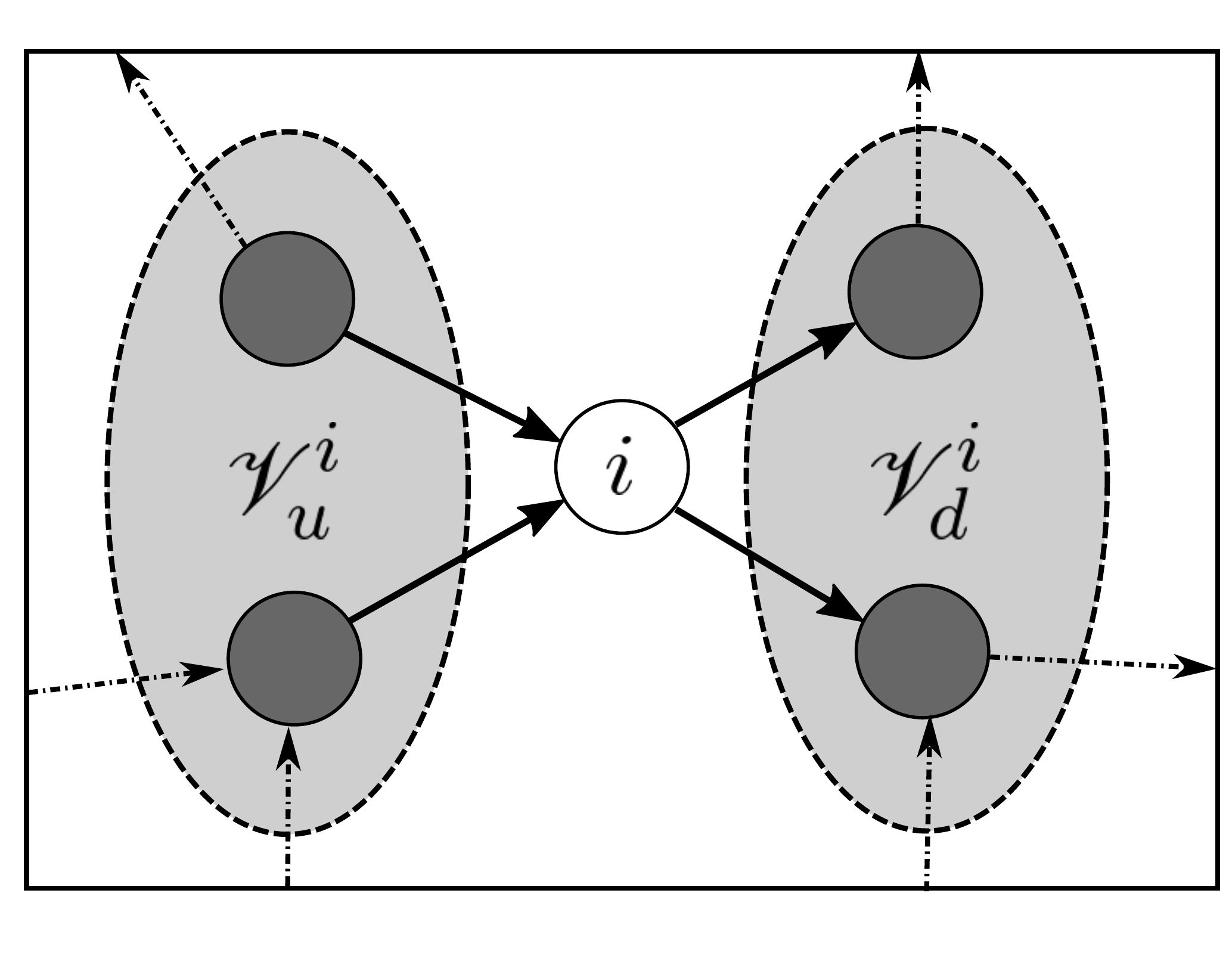}
		\caption{\label{fig:Notation Illustration} Illustration of upstream/downstream notation.}
	\end{figure}
	
	We can write the dynamics of the individual interacting subsystems as follows:
	\begin{gather} \label{eq:explicit system}
	x^i_{t+1} = a^i(\breve{x}^i_{t}, \breve{u}^i_{t}, \breve{u}^i_{t+1}), ~ ~i=1,...,N.
	\end{gather}
	where $a^i$ corresponds to the $i^{th}$ element in \eqref{eq:global system} and  $\breve{x}^i = \{x_j~|~ j \in \upstream{\Nodes}\}$ and $\breve{u}^i = \{u^j ~|~ j \in \upstream{\Nodes}\}$. 

	\subsection{Separable Optimization using ADMM} 
	 \label{sec:separability}
	
%
%
%

	Consider an optimization problem of the form,
	\begin{equation} \label{eq:generic_optimization}
	\underset{\theta}{\min} ~J(\theta), 
	\end{equation}
which may include constraints on $\theta$ via indicator functions appearing in $J$. The indicator function for the constraint $\theta \in \Theta$ is the function $\mathcal{I}_\Theta(\theta)$ which is zero for $\theta \in \Theta$ and infinite otherwise.
	
	\begin{definition}[Separable]
		The problem \eqref{eq:generic_optimization} is termed separable with respect to the partitioning $\theta = \{\theta^i ~|~ i=1,..,N\}$ if it can be written as $J(\theta) = \sum_{i=1}^N J^i(\theta^i)$. 
	\end{definition}

\noindent In this paper we encounter problems of the form:
		\begin{equation} \label{eq:admm_problem}
		\underset{\theta}{\min} ~\sum_{i=1}^N J_a^i(\theta^i_a) + \sum_{j=1}^M J_b^j(\theta^j_b), 
		\end{equation}
		where $\{\theta^i_a ~ | ~ i=1,...,N\}$ and $\{\theta^j_b ~ | ~ j=1,...,M\}$ are two different partitions of the same vector $\theta$. In our context, these partitionings correspond to the sets of upstream or downstream neighbors discussed in the previous section.
		 For such problems, the alternating directions method of multipliers (ADMM) can be applied \cite{boyd2011distributed}. We write \eqref{eq:admm_problem} as 
		\begin{align} \label{eq:separable_admm} 
		\min_{\theta, \phi}\ \  &  \sum_{i=1}^N J^i_a(\theta^i_a) + \sum_{j=1}^M J^j_b(\phi^j_b), \\
		s.t. \ \ &  \theta - \phi = 0. \nonumber  \nonumber
		\end{align}
		Applying ADMM results in iterations in which each step is separable with respect to the partition $\theta_a$ or $\theta_b$, and can thus be solved via distributed computing. For convex problems, ADMM is guaranteed to converge to the optimal solution \cite{boyd2011distributed}.
		
%


	\subsection{Problem Statement}
	To summarise, the main objective of this paper is as follows. Given state and input measurements $\{\tilde{x}_t, \tilde{u}_t~|~ t=1,..,T\}$, and a graph $\mathscr{G}$ describing the network topology, identify models \eqref{eq:explicit system} at each node such that: 
	\begin{itemize}
	  \item during the identification procedure, each subsystem only communicates with immediate (upstream and downstream) neighbours;
  \item convergence is guaranteed and least-squares equation error is small at each subsystem;
  \item model behavioural constraints such as contraction, monotonicity, and/or positivity can be guaranteed for the interconnected system \eqref{eq:global system}.
\end{itemize}

	
	\section{Convex Behavioral Constraints \label{sec:convex constraints} }

	In this section we develop a convex parametrization of models with contraction, monotonicity and/or positivity guarantees. As described in subsection \ref{sec:intro contraction}, jointly searching for a model \eqref{eq:global system} and contraction metric is non-convex.

	Following \cite{tobenkin2010convex, tobenkin2017convex}, we solve this problem by instead searching for models in the following implicit form:	
	\begin{equation} \label{eq:global_implicit_system}
	e(x_{t+1}, u_{t+1}) = f(x_t, u_t).
	\end{equation}
 	The differential dynamics of \eqref{eq:global_implicit_system} are: 
	 \begin{gather}
	 E(x_{t+1}, u_{t+1}) \delta_{x_{t+1}} = F(x_{t},u_{t}) \delta_{x_t} + K(x_{t},u_{t}) \delta_{u_t},  \label{eq:differential dynamics} 
	 \end{gather}
	 where $E = \frac{\partial e}{\partial x}$, $F = \frac{\partial f}{\partial x}$ and $K = \frac{\partial f}{\partial u}$. 

	\begin{definition}[Well-Posed]
		An implicit model of the form \eqref{eq:global_implicit_system} is termed  \textit{well-posed} if for every $x_t, u_t, u_{t+1}$ there is a unique $x_{t+1}$ satisfying \eqref{eq:global_implicit_system}.		
	\end{definition}
I.e., well-posedness means that $e(x,u)$ is a bijection with respect to its first argument, and implies the existence of an explicit model of the form \eqref{eq:global system} where $a = e^{-1}\circ f$. Furthermore, it implies  that for any initial condition $x_0$ and sequence of inputs $u_0,...,u_T$, there exists a unique trajectory $x_1,...,x_T$ satisfying \eqref{eq:global_implicit_system}.

	
	

	\subsection{ Stability and Monotonicity Constraints}
	
	In this section, we develop convex conditions on the implicit model \eqref{eq:global_implicit_system} that guarantee well-posedness, monotonicity, positivity, and contraction. The main result is the following:
	
	\begin{theorem}\label{thm: L1 contracting models}
		 A model of the form \eqref{eq:global_implicit_system} is:
		\begin{enumerate} [(a)]
			\item \textbf{well-posed} if there exists $\epsilon >0$ such that for all $(x,u)$,
			\begin{align}\label{eq:well-posed}
			E(x,u)+E(x,u)^T \succ \epsilon I,
			\end{align}
			\item \textbf{contracting} with rate $\alpha$ if (a) holds and there exists a matrix function $S(x,u):\R^n\times\R^m\to\R^{n\times m}$ such that for all $(x,u)$:
			\begin{align} 
			&-S(x,u) \le F(x,u) \le S(x,u),  \label{eq:slack_condition}\\ 			
			&\mathbf{1}^\top(\alpha E(x,u)-S(x,u))\ge 0, 	\label{eq:l1_contraction_polytone}
			\end{align}
			\item \textbf{monotone} if  (a) holds and for all $(x,u)$: \begin{align}\label{eq:monotonicity}
			F(x,u)\geq 0, \quad K(x,u)\geq 0, \quad E(x,u)\in \mathbb{M}^n,
			\end{align}
			
			\item \textbf{positive} if (c) holds and:\\
				\begin{equation}
			e(0) = f(0,0), \label{eq:origin_equilibria}
			\end{equation}
			
			\item \textbf{contracting and monotone} if (a) and (c) hold, and for all $(x,u)$
			\begin{align} \label{eq:l1 contraction condition} 
			&\mathbf{1}^\top(\alpha E(x,u)-F(x,u))\geq 0.
			\end{align}
			Positivity is also enforced if \eqref{eq:origin_equilibria} holds.
			
		\end{enumerate}
	\end{theorem}
	
	\begin{proof}
		See appendix \ref{appendix:thm1}.
	\end{proof}

	We refer to the stability conditions in Theorem \ref{thm: L1 contracting models} (b) or (e) as $\ell_1$ contraction conditions as they ensure contraction using a state dependent weighted $\ell_1$ norm of the differentials: $V(t, x, \delta) = |E(x,u)\delta|_1$, noting that for the purpose of contraction analysis the exogenous input $u$ can be considered as a time-variation.

	\begin{remark} 
		Theorem \ref{thm: L1 contracting models} requires an exponential contraction rate $\alpha$ to be specified. A weaker form of incremental stability can also be imposed by replacing \eqref{eq:l1_contraction_polytone} with
		\begin{gather}
		\mathbf{1}^\top(E(x,u)-S(x,u))\ge \mu \bm{1}^\top
		\end{gather}
		for some $\mu > 0$, and similarly for \eqref{eq:l1 contraction condition}. This implies that $\sum_{t=0}^\infty |x^a_t-x^b_t|_1<\infty$, following a line of reasoning similar to \cite{tobenkin2017convex}.
	\end{remark}


	\subsection{Model Parametrizations} \label{sec:model parametrizations}
		
	
	As formulated above, Theorem \ref{thm: L1 contracting models} applies to models represented by the infinite dimensional space of continuously differentiable functions $e$ and $f$. In practice, these functions are usually parametrized by a finite-dimensional vector. In this section we briefly discuss some common model parametrizations and how the constraints can be enforced.
	
	For linear models, \eqref{eq:well-posed} is a semidefinite constraint, \eqref{eq:slack_condition}-\eqref{eq:l1 contraction condition} are linear and can be enforced using semidefinite programming. Furthermore, if $E$ is diagonal, then \eqref{eq:well-posed} is also linear and the model set is polytopic.
	
	If the functions $e$ and $f$ are  multivariate polynomials or trigonometric polynomials, then the constraints can be enforced using sum of squares programming \cite{parrilo2003semidefinite,megretski2003positivity}.
	
	The model set \eqref{eq:global_implicit_system} also contains a class of recurrent neural networks with slope-restricted, invertible activation functions. In this case, $e(x)$ is the inverse of the activation functions, $f(x,u)$ is affine, and \tb{simulation of the explicit model $a=e^{-1}\circ f$ yields the equation of a standard recurrent neural network \cite{elman1990finding}.
	The }conditions in Theorem \ref{thm: L1 contracting models} (b) or (d) then correspond to diagonal dominance conditions on the weight matrices which can be enforced via linear constraints.
	
	Finally, if the requirement for global verification of these properties is relaxed, then these constraints can be applied pointwise for arbitrary parametrizations $e$ and $f$, which amount to linear and semidefinite constraints if $e$ and $f$ are linearly parametrized.

	\section{Distributed Identification \label{sec:distributed algorithm}}
	

	In this section we consider the problem of distributed identification of networked systems with the behavioral constraints introduced in Theorem \ref{thm: L1 contracting models}.
	First, we propose a particular structure for \eqref{eq:global_implicit_system} for which the constraints in Theorem \ref{thm: L1 contracting models} are separable. We then propose an objective function that is separable (with respect to a different partition). Finally we propose an algorithm for fitting the proposed models that requires only local communication between subsystems at each step.
	
	\subsection{Distributed Model}\label{sec:model structure}
	


	We propose the following model structure for distributed identification, in which $e$ depends only on local states and inputs, and $f$ is a summation of nonlinear functions of states and inputs from upstream neighbours:
	\begin{gather}\label{eq:implicit model}
	e^i(\argse{t+1}) = \sum_{j \in \upstream{\Nodes}} f^{ij}(x^j, u^j ). 
	\end{gather}
	Models of the form \eqref{eq:implicit model} are widely used for statistical modelling, and are referred to as generalized additive models (GAMs) \cite{hastie1990generalized}. This class of models also includes linear systems, and a class of recurrent neural networks.
	We assume that each of the functions $e^{i}: \mathbb{R}^{n_i} \times \mathbb{R}^{m_i} \mapsto \mathbb{R}^{n_i}$ and $f^{ij}: \mathbb{R}^{n_j} \times \mathbb{R}^{m_i} \mapsto \mathbb{R}^{n_i}$ are linearly parametrized by $\theta_e^{i}$ and $\theta_f^{ij}$ respectively. 

	
		We define two partitions of the model parameters; the sets of upstream and downstream parameters. These are denoted $\upstream{\theta} = \{\theta^{i}_e, \theta^{ij}_f| j \in \upstream{\Nodes}\}$ and  $\downstream{\theta} = \{\theta^{i}_e, \theta^{ji}_f| j \in \downstream{\Nodes}\}$ respectively. Objective functions, constraints and optimization problems are called \textit{upstream-separable} or \textit{downstream-separable} if they are separable with respect to these partitions. Upstream and downstream separable optimization problems are closely related to the column-wise and row-wise separable optimization problems used in \cite{wang2018separable}.

		For the parametrization \eqref{eq:implicit model}, the differential dynamics have a sparsity pattern determined by the network topology. In particular, the  $(i,k)^{th}$ block of $F$ is:
		\begin{gather*}
		F^{ik} = \parDiff{}{x^k} \sum_{j\in \upstream{\Nodes}} f^{ij}(x^j, u^j) = \begin{cases}
		\parDiff{f^{ik}}{x^k}, & k \in \upstream{\Nodes} \\ 0, & k \notin \upstream{\Nodes}
		\end{cases}.  
		\end{gather*}
		and $E$ is block diagonal.
		This means $F^{ik}$ depends only on parameters $\theta_f^{ik}$ and the block $E^{ii}$ depends only on $\theta_e^{i}$. As each block  of $E$ and $F$ has an independent parametrization, functions of disjoint sets of elements of $E$ or $F$ will be separable.

	\subsection{Convex Bounds for Equation Error}	\label{sec:LREE}
	
		In Section \ref{sec:model structure} we propose a convex set of implicit models. However, this approach shifts the convexity problem from the model set to the objective function as equation error \eqref{eq:explicit_equation_error}, s.t. $a =  e^{-1} \circ f$, is no longer convex in the model parameters.
		
		One approach might be to minimize the implicit equation error 
		\begin{equation}\label{eq:implicit equation error}
		J_{iee} = \sum_{t=1}^{T-1} | e(\tilde{x}_{t+1}, \tilde{u}_{t+1}) - f(\tilde{x}_{t}, \tilde{u}_{t})|^2
		\end{equation} 
		as a surrogate for equation error. This approach however, strongly biases the resulting model and leads to poor performance \cite{umenberger2019convex}.  Instead we use the convex upper bound for equation error proposed in \cite{umenberger2019convex}, which is based on Lagrangian relaxation.

	The least-squares equation error \eqref{eq:explicit_equation_error} for the implicit model \eqref{eq:global_implicit_system} is: 
	\begin{alignat}{4}
	&  \min_{\theta, x_2,...,x_{T}} & \quad &  J_{ee}(\theta) = \sum_{t=1}^{T-1} |x_{t+1} - \tilde{x}_{t+1}|^2  \label{eq:LREEOpt}\\
	&  \ \quad\centering s.t.                  &		  &	  e(x_{t+1}, \tilde{u}_{t+1}) = f(\tilde{x}_{t} , \tilde{u}_{t}), \quad \forall t = 1, ..., T-1. \nonumber 
	\end{alignat}
	Note that this problem is not jointly convex in $x_{t+1}$ and $\theta$. The following convex upper bound was proposed in \cite{umenberger2019convex}:	
	\begin{multline} \label{eq:Lagrangian Relaxation}
	J_{ee} \leq \hat{J}_{ee}({\theta}) = \sum_{t=1}^{T-1} \sup_{x_{t+1}} \bigg\{ |x_{t+1} - \tilde{x}_{t+1}|^2 \\ - 2 \lambda(x_{t+1})^\top (  e(x_{t+1},\tilde{u}_{t+1}) - f(\tilde{x}_{t} , \tilde{u}_{t})) \bigg\},
	\end{multline}	
	where $\lambda_{t}(x_{t+1}) = x_{t+1} - \tilde{x}_{t+1}$ is a Lagrange multiplier. The  function \eqref{eq:Lagrangian Relaxation} is convex in $\theta$ as it is the supremum of an infinite family of convex functions \cite{tobenkin2017convex}. 
	
	For our parametrization \eqref{eq:implicit model}, $E$ is block diagonal which then implies that \eqref{eq:Lagrangian Relaxation} is upstream separable so it can be written as 
	\begin{gather}\label{eq:Distributed LREE}
	\hat{J}_{ee}(\theta) = \sum_{i=1}^{N} \hat{J}_{ee}^i(\upstream{\theta}), 
	\end{gather}
	where
	\begin{multline*}
	\hat{J}^i_{ee}({\upstream{\theta}}) = \sum_{t=1}^{T-1} \sup_{x^i_t} \bigg\{ |x^i_{t+1} - \tilde{x}^i_{t+1}|^2 \\ - 2 (x^i_{t+1} - \tilde{x}^i_{t+1})^\top \bigg( e^i(x^i_{t+1},\tilde{u}^i_{t+1}) - \sum_{j \in \upstream{\Nodes}} f^{ij}(\tilde{x}^j_{t} , \tilde{u}^j_{t})\bigg) \bigg\}.
	\end{multline*}	
	
	The evaluation of $\hat{J}^i_{ee}$ is not trivial as it involves the calculation of the supremum of a non-linear multivariate function. In this work we linearise \eqref{eq:Distributed LREE} with respect to $x^i_t$ and solved for the supremum of the resulting concave quadratic function, giving:
	
	\begin{equation} \label{eq:linearized lree sup}
	\hat{J}^i_{ee}(\upstream{\theta}) \approx \lree^i(\upstream{\theta}) = \sum_{t=1}^{T-1} {\epsilon^i_t}^\top (E^i_t + {E^i_t}^\top - I)^{-1} \epsilon^i_t,
	\end{equation}
	where $\epsilon^i_t =  e^i(x^i_{t+1},\tilde{u}^i_{t+1}) - \sum_{j \in \upstream{\Nodes}} f^{ij}(\tilde{x}^j_{t} , \tilde{u}^j_{t})$ is the implicit equation error and  $E^i(x^i, u^i) = {\partial e^i}/{\partial x^i}$ and $E^i_t = E^i(\tilde{x}^i_t, \tilde{u}^i_t)$.
	The cost function \eqref{eq:linearized lree sup} can be optimized via a semidefinite program. Alternative methods for minimizing LREE can also be found in \cite{umenberger2019convex}.

	\subsection{Alternating Directions Method of Multipliers (ADMM) \label{sec:ADMM}}
	
	In Section \ref{sec:model structure} we introduced a model set for which the constraints in Theorem \ref{thm: L1 contracting models} are downstream separable and in Section \ref{sec:LREE} we introduced an upstream separable objective function. Note however, that the constraints and objective are not jointly separable with respect to the same partition.  We use ADMM to solve this problem.

	We \tb{now } develop the algorithm for the case where \eqref{eq:global_implicit_system} is well-posed, monotone and contracting, however, a parallel construction without monotonicity or contraction constraints introduces no additional complexity. Consider the following set of parameters
	\begin{equation}\label{eq:theta_ml1}
		\mc{} = \{\theta ~ | ~ \eqref{eq:well-posed}, \eqref{eq:monotonicity}, \eqref{eq:origin_equilibria}, \eqref{eq:l1 contraction condition}\}.
	\end{equation}
	Applying  ADMM as discussed in Section \ref{sec:separability} to the problem $\min_{\theta \in \mc{}}  \hat{J}_{ee}$ gives the following iteration scheme
	for iteration $k$:
	\begin{flalign} 
	&\theta(k+1) = \argmin_{\theta}  {\hat{J}_{ee}}(\theta) + \frac{\rho}{2}||\theta  - \phi(k) + v(k)||^2, \label{ADMM iterates 1.1} &\\
	&	\phi(k+1) = \argmin_{\phi \in \mc{}}  \frac{\rho}{2}||\theta(k+1) - \phi -  v(k)||^2,  ~\label{ADMM iterates 1.2}& \\
	&	v(k+1) = v(k) - \theta(k+1) + \phi(k+1). & \label{ADMM iterates 1.3} 
	\end{flalign}
	for $\rho > 0$.	
	
	 When using a GAM structure \eqref{eq:implicit model}, we have the following result:
	\begin{proposition} \label{thm:admm_is_separable}
		For the model structure \eqref{eq:implicit model}, the ADMM iteration \eqref{ADMM iterates 1.1} separates into $N$ upstream-separable optimization problems of the form \eqref{ADMM iterates 2.1} and the ADMM iteration \eqref{ADMM iterates 1.2} separates into $N$ downstream-separable optimization problems of the form \eqref{ADMM iterates 2.2}.
	\end{proposition}
	\begin{proof}
		See Appendix \ref{appendix:admm_is_separable}.
	\end{proof}

In particular, the ADMM approach corresponds to performing the following iterations locally at each node $i=1,...,N$:
	\begin{align}
	\upstream{\theta}(k+1) & = \begin{multlined}[t]
	\argmin_{\upstream{\theta}} \hat{J}^i_{ee}(\upstream{\theta})  \\ + \frac{\rho}{2}||\upstream{\theta} - \upstream{\phi}(k)  + \upstream{v}(k)||^2, 
	\end{multlined} \label{ADMM iterates 2.1}    \\
	\downstream{\phi}(k+1) & = 
	\begin{multlined}[t]
	\argmin_{\downstream{\phi}} 
	\mathcal{I}_{\mc{}}(\downstream{\phi})   \\ +  \frac{\rho}{2}||\downstream{\theta}({k+1}) - \downstream{\phi}  + \downstream{v}(k)||^2,
	\end{multlined} \label{ADMM iterates 2.2} \\
	\upstream{v}(k+1) & = 
	\upstream{v}(k) - \upstream{\theta}(k+1) + \upstream{\phi}(k+1). \label{ADMM iterates 2.3}
	\end{align}

%

	
	The distributed algorithm is listed in Algorithm \ref{alg:admm}. The steps \eqref{ADMM iterates 2.1} and \eqref{ADMM iterates 2.2} require access to the upstream and downstream parameters respectively. These can be solved by the nodes in the graph, however, communication between both upstream and downstream parameters is necessary between steps.
	The update \eqref{ADMM iterates 2.3} is trivially separable and can be solved as either an upstream or downstream separable problem.

	\begin{algorithm}
		\SetAlgoLined
		\KwResult{ $\phi$}
		Initialize $\rho>0$\;
		Initialize: 
		${\theta}(0)$, ${\phi}(0)$, ${v}(0)$\;
		
		\For{ $k = 1,...$}{
			\For{ $i = 1,...,N$}{
				Get: $\{\breve{x}^i_t\}_{t=1}^T$\;
				
				Compute $\upstream{\theta}(k+1)$ using \eqref{ADMM iterates 2.1}\;
				Send $\upstream{\theta}(k+1)$ to upstream neighbours\;
				
				Compute $\downstream{\phi}(k+1)$ using \eqref{ADMM iterates 2.2}\;
				Send $\downstream{\phi}(k+1)$ to downstream neighbours\;
				
				Compute $\upstream{v}$ using \eqref{ADMM iterates 2.3}\;
				Send $\upstream{v}(k+1)$ to upstream neighbours\;
			}
		}
		\caption{\label{alg:admm}Distributed Algorithm}
	\end{algorithm}	
	
	%
	%
	
	Termination of ADMM after a finite number of iterations means that the two parameter vectors $\theta$ and $\phi$ will disagree. For this reason, we take $\phi$ as the solution to ensure that the well-posedness, monotonicity and contraction constraints \eqref{eq:well-posed}, \eqref{eq:monotonicity}, \eqref{eq:l1 contraction condition} are satisfied.

	\section{Discussion} \label{sec:Discussion}
	
	\subsection{Conservatism of the Separable Model Structure} 	\label{sec:Discussion-model set}
	
	We have proposed searching over the model set \eqref{eq:implicit model} with $\theta \in \Theta_{m\ell_1}$ \eqref{eq:theta_ml1}, and it is important to understand which systems may fall into this model set. A particular question of interest is whether there are contracting and monotone systems which \textit{cannot} be represented by this structure, and there are two main reasons why this may occur: the separable structure of the model \eqref{eq:implicit model}, and the assumption of a separable contraction metric in condition \eqref{eq:l1 contraction condition}.

		An exact characterization of the functions functions that be approximated via the GAM structure \eqref{eq:implicit model} is difficult to give, however, they have widely applied in statistical modelling, see \cite{hastie1990generalized} for details. Note that while the functions in the implicit system \eqref{eq:implicit model} are additive, the resulting explicit system \eqref{eq:explicit system} may not be. For example, the scalar functions $e(x) = \sqrt{x}$ and $f(x, y) = (x + y)$. Both $e$ and $f$ are additive; however, the function $e^{-1}\circ f(x,y) = {x}^2 + 2xy + y^2$ is not. 
		
		Conservatism  may also be introduced by the assumption of a separable contraction metric. 
		For the case of linear positive systems, it is has been shown that the existence of a separable Lyapunov functions is both necessary and sufficient \cite{berman1994nonnegative}. This means that $\Theta_{m \ell_1}$ contains all positive linear systems \cite{umenberger2016scalable}:
		\begin{theorem} \label{thm:discrete_time_positive_linear}
			For the system \eqref{eq:implicit model}, if $e$ and $f$ are affine in $(x,u)$, then the model set characterised by \eqref{eq:well-posed}, \eqref{eq:monotonicity} and \eqref{eq:l1 contraction condition} is a parametrization of all stable, discrete-time, positive linear systems.
		\end{theorem}
		\begin{proof}
			See Appendix \ref{appendix:thm2}.
		\end{proof}
		Things are more complicated for nonlinear monotone systems. Separable contraction metrics have been shown to exist for certain classes of monotone systems \cite{manchester2017existence} and separable weighted $\ell_1$ contraction metrics have been used for the analysis of monotone systems \cite{como2015throughput,coogan2019contractive}. For incrementally exponentially stable systems, it has been shown that the existence of weighted $\ell_1$ contraction metrics, are necessary and sufficient \cite{kawano2019contraction}, however the state-dependant weighting depends on the all system states and is therefore not separable in the sense we use. To the authors' knowledge, a complete characterisation of the class of contracting monotone systems that admit separable metrics is still an open problem.
		 
		

	\subsection{Consistency} \label{sec:Discussion-consistency}
	
		It has be previously noted that system identification approaches that guarantee stability lead to a bias towards systems that are too stable \cite{maciejowski1995guaranteed,lacy2003subspace,manchester2012stable}. Empirical evidence suggests that for methods based on Lagrangian relaxation \cite{umenberger2018specialized,umenberger2019convex} this bias is smaller.
		
		
		There are a number of situations that lend themselves towards consistent identification. Firstly, consider the situation where we have noiseless state and input measurements produced by a model with $\theta^* \in \Theta_{m \ell_1}$ such that $J_{ee}(\theta^*) = 0$. Then we also have $\lree(\theta)^* = 0$ so the bound is tight and LREE recovers the true minimizer of equation error.
		
		Now, consider the situation where the unconstrained minimizer of equation error \eqref{eq:explicit_equation_error}, is a monotone, additive function that is contracting in the identity metric. That is, for the function $a_{\phi^*}(x,u)$ where $\phi^* = \argmin J_{ee}(\phi)$, the following hold:
		\begin{enumerate}
			\item $a_{\phi^*}(x,u)$ is additive so that \eqref{eq:explicit system} can be written as $a^i(x, u) = \sum_{j \in \upstream{\Nodes}}a^{ij}(x^j, u^j)$, \label{ass:additive}
			\item  $\bm{1}^\top(\alpha I - A(x,u))\geq 0$, \label{ass:contracting}
			\item $A(x,u)\geq0$. \label{ass:monotone}
		\end{enumerate} where $A = \parDiff{a}{x}$. Then, optimizing \eqref{eq:linearized lree sup} returns the same solution as the unconstrained least squares minimizer of $J_{ee}$.
		\begin{proposition}
			Consider models of the form \eqref{eq:implicit model} with $e_\theta(x,u) = Ex$ and $f_\theta(x,u) = a_{\phi^*}(x,u)$ for some $\theta$. If properties \ref{ass:additive}, \ref{ass:contracting}, \ref{ass:monotone} hold for $a_{\phi^*}(x,u)$ where $\phi^* = \argmin J_{ee}(\phi)$, then for $\theta^* = \underset{\theta \in \Theta_{m \ell_1}}{\argmin} ~ \lree(\theta)$, we have $a_{\phi^*}(x,u) = e_{\theta^*}^{-1}f_{\theta^*}(x,u)$.
		\end{proposition}
		\begin{proof}
			Our proof mirrors that of \cite[Sec. IV Proposition 1]{umenberger2019convex}.
		\end{proof}

	\subsection{Iteration Complexity of Distributed Algorithm}
		In this section, we investigate the computational complexity of each step in the distributed algorithm.
		In general, the complexity depends on the model parametrization used, however, we limit our discussion to the case where the models are parametrized by polynomials and the constraints are enforced using sum of squares programming.

		The first step, \eqref{ADMM iterates 2.1}, is a semi-definite program and can be solved using standard solvers. If no structural properties are exploited, a primal-dual interior point method (IPM), would require $\mathcal{O}\left[\max \{ n_{\upstream{\theta}}^3,~ n_{\upstream{\theta}}  {{n}_i}^3 ,~ n_{\upstream{\theta}}^2  {n}_i^2 \}\right]$ operations per iteration per node \cite{liu2009interior}, where $n_{\upstream{\theta}}$ is the number of upstream free parameters .
	
	The second step, \eqref{ADMM iterates 2.2}, is a sum-of-squares problem that can solved as a semi-definite program. If $e$ and $f$ both have degree $2d$, then the size of Gram matrix corresponding to \eqref{eq:l1 contraction condition} for the additive model \eqref{eq:implicit model} is $ p = 1 + \sum_{j\in \downstream{\Nodes}} \left[ {n_j + m_i + d \choose d} - 1\right]$. Solving \eqref{ADMM iterates 2.2} using a primal-dual IPM requires approximately $\mathcal{O}\left[\max \{ n_{\downstream{\theta}}^3,~ n_{\downstream{\theta}} p^3 ,~ n_{\downstream{\theta}}^2  p^2 \}\right]$ operations per iteration per node \cite{liu2009interior}, where $n_{\downstream{\theta}}$ is the number of downstream free parameters.
	
	If  a local computational resource is associated with each node in the network, and the number of neighbours for each node satisfies a uniform bound, then the time taken for each iteration will not increase with the number of nodes. However, computation time will grow quickly with the number neighbours, the size of the local states and the degrees of the polynomials used in the model. 
	
	\subsection{Other Quality of Fit Criteria}
	
	Lagrangian relaxation of least-squares equation error was chosen as it is convex, upstream separable, \tb{quick to compute}, and leads to a simple implementation of ADMM. 
	\tb{Any method that treats neighbouring states as exogenous inputs will be upstream separable. However, any such approach will also be susceptible to instability due to the introduction of new feedback loops via the network topology, even if it guarantees stability of the local models. 
	Consequently,  one can similarly apply any convex quality of fit criteria such us convex upper bounds on simulation error \cite{tobenkin2017convex,umenberger2018specialized} and still guarantee convergence of ADMM. }
	Alternatively, a non-convex quality of fit criteria like simulation error can be used at the expense of ADMM's convergence guarantees.
	
	If a model structure does not permit distributed identification, the conditions proposed in Section \ref{sec:convex constraints} can still be used to ensure stability and/or monotonicity. Joint convexity of the model set and stability constraints is still an important as it simplifies constrained optimization allowing for the easy application of penalty, barrier or projected gradient methods \cite{revay2019contracting}.

	\section{Numerical Experiments \label{sec:results}}

	In this section we present numerical results exploring the scalability and identification performance the proposed approach. 

	This section is structured as follows:
	first, we look at the identification of positive linear systems, and explore the computational complexity of the $\ell_1$ and $\ell_2$ contraction conditions;
	we then explore the consistency of fitting nonlinear models when the true system lies in the model set, essentially analysing the effect of convex bound on equation error;
	finally, we apply the method to the identification of a (simulated) nonlinear traffic network. The traffic network does not lie in the model set so only an approximate model can be identified. We explore the regularising effect of the model constraints and scalability of the method to large networks.

	Previous methods for the identification of models with stability guarantees have ensured contraction using a quadratic metric \cite{tobenkin2017convex,umenberger2018specialized,umenberger2019convex}. Contraction is implied by the following semidefinite constraint:
	\begin{gather} \label{eq:L2 Contraction}
	W(x,u,\theta) \succeq 0 \ \ \ \ \forall (x,u),\\
	W(x,u,\theta)=\begin{bmatrix}
	E(x,u) + E(x,u)^\top  - P - \eta I & F(x,u)^\top \\ F(x,u) & P
	\end{bmatrix} \nonumber
	\end{gather}
	where $P \in \mathbb{S}^{n\times n}, P \succ 0$, $\eta > 0$. We refer to \eqref{eq:L2 Contraction} as an $\ell_2$ contraction condition as it implies the contraction conditions \eqref{eq:contraction_metric} with  a state dependent weighted $\ell_2$ norm of the differentials $V = \delta_{x_t}^\top E(x_t, u_t)^\top P^{-1} E(x_t, u_t) \delta_{x_t}$.
	
	
	We will make future reference to the following convex sets of parameters, in addition to $\theta_{ml_1}$ defined in \eqref{eq:theta_ml1}:
	\begin{gather*}
	~~	\Theta_{u} = \{ \theta ~|~ \eqref{eq:well-posed}, \eqref{eq:origin_equilibria} \}, ~ ~ \Theta_{m} = \{ \theta ~| ~ \eqref{eq:well-posed}, \eqref{eq:monotonicity}, \eqref{eq:origin_equilibria}\} \\ \Theta_{m \ell_2} = \{ \theta~ | ~ \eqref{eq:monotonicity}, \eqref{eq:origin_equilibria}, \eqref{eq:L2 Contraction} \}
	\end{gather*}
	Here the subscripts refer to the following properties:
	\begin{itemize}
		\item $m \ell_1$ -  Monotone $ \ell_1$ contracting models i.e. $\theta \in \mc{}$,
		\item $m$ - Monotone models i.e. $\theta \in \m{}$,
		\item $u$ - Models that are not constrained to be contracting or monotone i.e. $\theta\in 
		\uncon{}$, 
		\item $m \ell_2$ - Models that are monotone and contracting in $ \ell_2$, i.e. $\theta \in \elltwo{} $,
		\end{itemize}
	All functions $e^{i},$ and $f^{ij}$ are polynomials in all monomials of their arguments up to a certain degree. 
	
	As a baseline for comparison, we will also compare to models denoted $Poly$, with explicit polyonomial models \eqref{eq:global system} fit by least-squares without any separable structure imposed. \tb{We will also compare to standard wavelet and sigmoid Nonlinear AutoRegressive with Exogenous input (NARX) models implemented as part of the Matlab system identification toolbox.}
	
	For the implicit models, the model class prefix is followed by the degrees of the polynomials in $e$ and $f$ in parenthesis. For example, the notation $u(3,5)$ refers to unconstrained models with $e$ having degree $3$ and $f$ having degree $5$.
	For the explicit polynomial models $Poly$, the degree used follows in parenthesis, so $Poly(5)$ are explicit polynomial models of degree 5 in all arguments. 
	
	\tb{
		The NARX models were fit at each node using the regressors $(\breve{x}^i_t, \breve{u}^i_t, \breve{u}^i_{t+1})$. 
		The wavelet NARX models were set to automatically choose the number of basis functions and the sigmoid NARX models were set to use $10$ basis functions.
		The focus for each model was set to produce the best performance. For the wavelet network, we used a focus on simulation and for the sigmoid network, we used a focus on prediction.
	}

	The constraints \eqref{eq:well-posed}, \eqref{eq:monotonicity}, \eqref{eq:origin_equilibria}, \eqref{eq:l1 contraction condition} and \eqref{eq:L2 Contraction} are enforced using sum of squares programming \cite{parrilo2003semidefinite}.	
	All programs are solved using the SDP solver MOSEK with the parser YALMIP \cite{lofberg2004yalmip} on a standard desktop computer (intel core i7, 16GB RAM).


	\subsection{Identification of Linear Positive Systems}
	In this subsection we study the scalability of the proposed method for the identification of linear positive  systems. 
	
	We compare the computation time using the proposed $\ell_1$ contraction constraint to a previously proposed $\ell_2$ contraction constraint (i.e. quadratic Lyapunov function). 
	Note that for linear systems, the model sets $m\ell_1$ and $m\ell_2$ both are parameterizations of all stable positive linear systems so no difference in quality of fit is observed.

	
	\subsubsection{Scalability of Separable Linear and Quadratic Metrics} \label{sec:ID_of_linear_system}
	
	We illustrate the difference in scalability between the models $m \ell_1{(1,1)}$ and $m \ell_2{(1,1)}$.
	Each experimental trial consists of the following steps:
	\begin{enumerate}
		\item A stable positive system with state dimension $n_x$ is randomly generated using Matlab's \cfont{rand} function; $A\in\real^{n_x\times n_x}$ has a banded structure with band width equal to $9$. Stability was ensured by rescaling $A$ to have a spectral radius of $0.95$.
		\item The system is simulated for $T=10^4$ time steps; $\tilde{x}_{1:T}$ is obtained by adding white noise to the simulated states at SNR equal to 40dB. 
		\item This process is repeated 5 times for each $n_x$.
	\end{enumerate}

	\begin{figure}
		\centering
		\includegraphics[width=\linewidth, trim = {0cm, 0cm, 0.5cm, 0.5cm}, clip ]{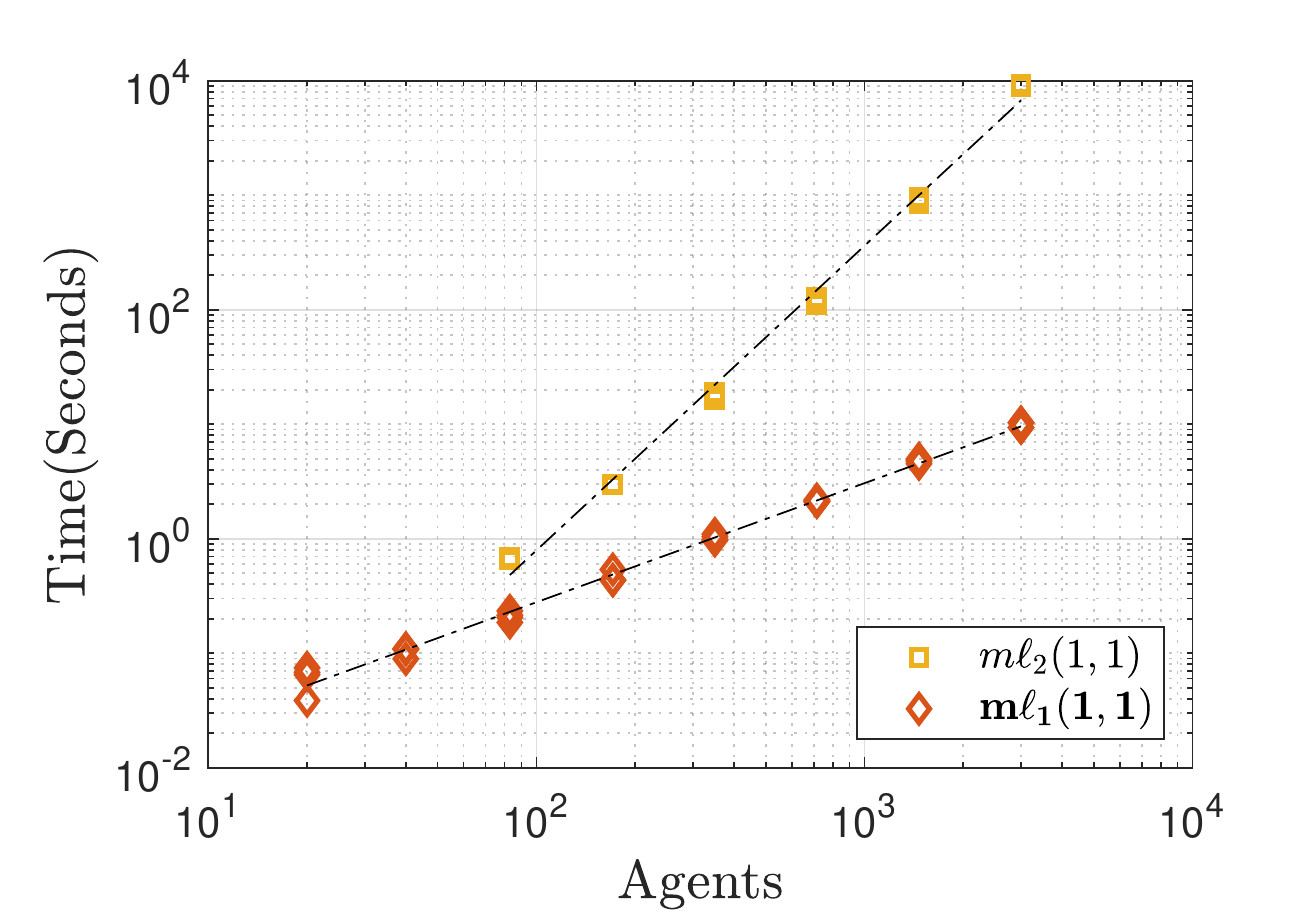}
		
		\caption{ \label{fig:linear runtimes} Computation time as function of system size. The slopes of the lines of best fit are: $m\ell_2(1,1)$ - 2.66, $m\ell_1(1,1)$ - 1.04 . }
	\end{figure}
	
	The time taken to solve each optimization problem is shown in Fig. \ref{fig:linear runtimes}. Here, we see a significant improvement in the computational complexity from approximately cubic growth for $m \ell_2$ to linear growth for $m \ell_1$. The networked approach allows us to solve stable identification problems with at least $3000$ states.
	
	Note that no explicit attempts to exploit the sparsity of the system were made; use of solvers and parsers designed to exploit sparsity could improve performance, especially for the SDPs associated with the LMI parametrization, e.g. \cite{andersen2010implementation}.

	\subsection{Identification of Nonlinear Models}
	In this section we study the consistency of fitting nonlinear implicit models via the LREE bound on equation error. In Section \ref{sec:Discussion-consistency} we saw that in the noiseless case, optimization of LREE will return the true model parameters. We will now explore the effect of introducing noise on the model estimates. 
	The experiments in this section can be seen to supplement those in \cite[Sec. IV]{umenberger2019convex} which studied the effects of noise and model stability on consistency in the linear setting.
	
	We generate models $a^*(x,u)$ by sampling a parameter vector $\theta$ and then projecting onto the set $\mc{}$. The models have degree 3, state size $n=2$ and $m=1$.
	We then generate training data with $T$ samples by randomly sampling $(\tilde{x}_t,\tilde{u}_t)$ from the uniform distribution on $[0, 1]$ and generated noisy measurements of $x_{t+1}$ by $\tilde{x}_{t+1} = a^*(\tilde{x}_t,\tilde{u}_t) + v_t$, where $v_t$ is normally distributed noise with a specified Signal to Noise Ratio (SNR). Models $a(x,u)$ are then trained by minimizing $\lree$ with $\theta \in \mc{}$ and performance measured using Normalized Equation Error (NEE):
	\begin{equation}\label{eq:NEE}
	\text{NEE} = \frac{|a(x,u) - a^*(x,u)|^2_2}{|a^*(x,u)|^2_2}
	\end{equation}
	where $a(x,u)$ is the identified dynamic model and $a^*(x,u)$ is the true 
	where $|f(x)|_2 = \int_{x\in \mathcal{D}} |f(x)|^2 dx$ is the sample estimate of the 2-norm of the function $f$. 
		
	\begin{figure}[]
		\centering
		\includegraphics[trim = {0cm 0.0cm 1cm 2cm}, clip,width=0.8\linewidth]{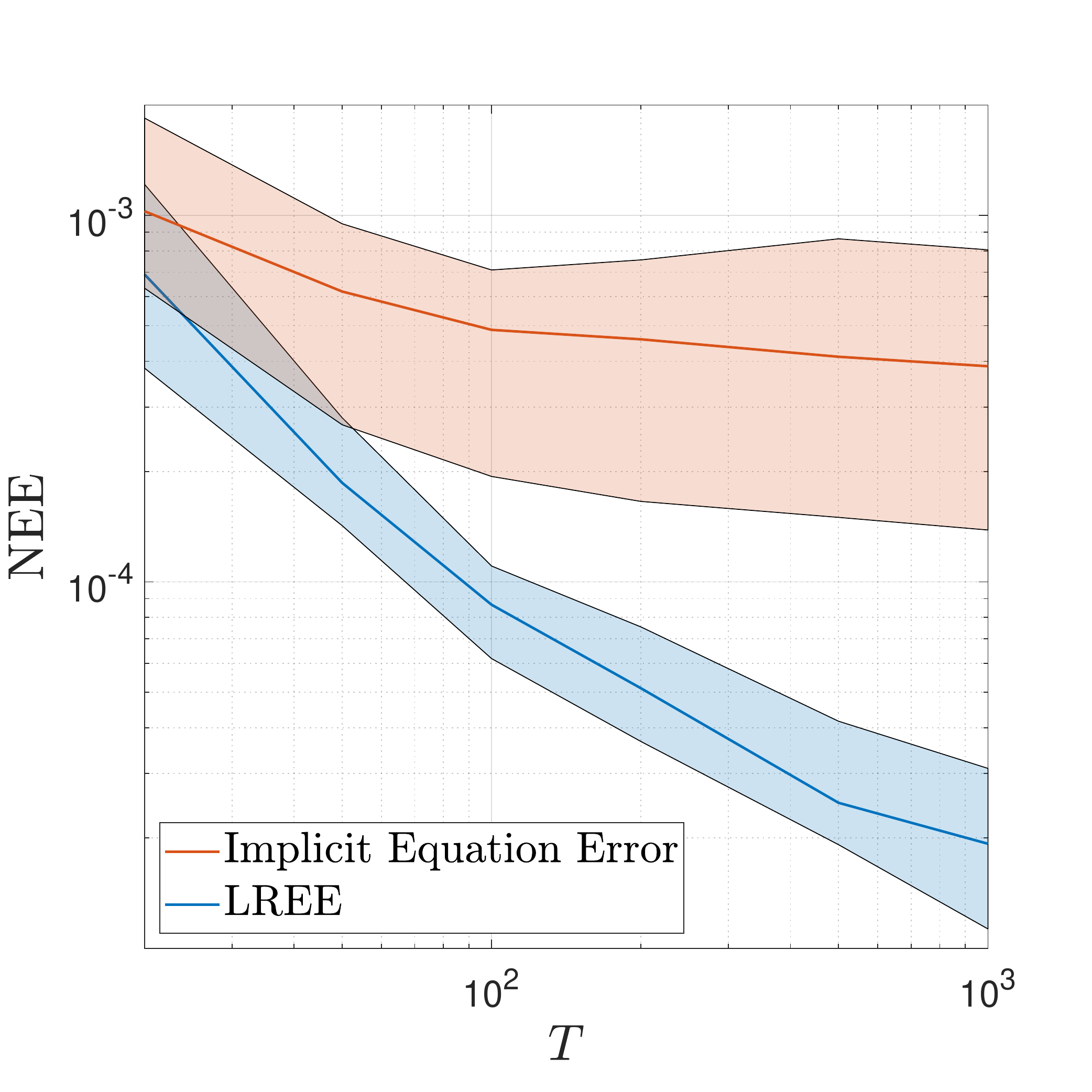}
		\caption{Comparison of implicit equation error and LREE: Normalized equation error versus number of training data points. The training data has gaussian noise with $SNR=30 \text{dB}$. For each method, the central line shows the median NEE for 50 model realizations and the shaded region shows the upper and lower quartiles.\label{fig:iee_vs_lree}}
	\end{figure}
	
	In Figure \ref{fig:iee_vs_lree}, we have plotted the NEE that results from fitting models from $m\ell_1(3,3)$ by optimizing LREE \eqref{eq:linearized lree sup} and implicit equation error \eqref{eq:implicit equation error}. 
	We can see that LREE provides a much better fit than implicit equation error, especially as the number of data points increases.

	\begin{figure}[]
		\centering
		\includegraphics[trim = {0cm 0.0cm 1cm 2cm}, clip,width=0.8\linewidth]{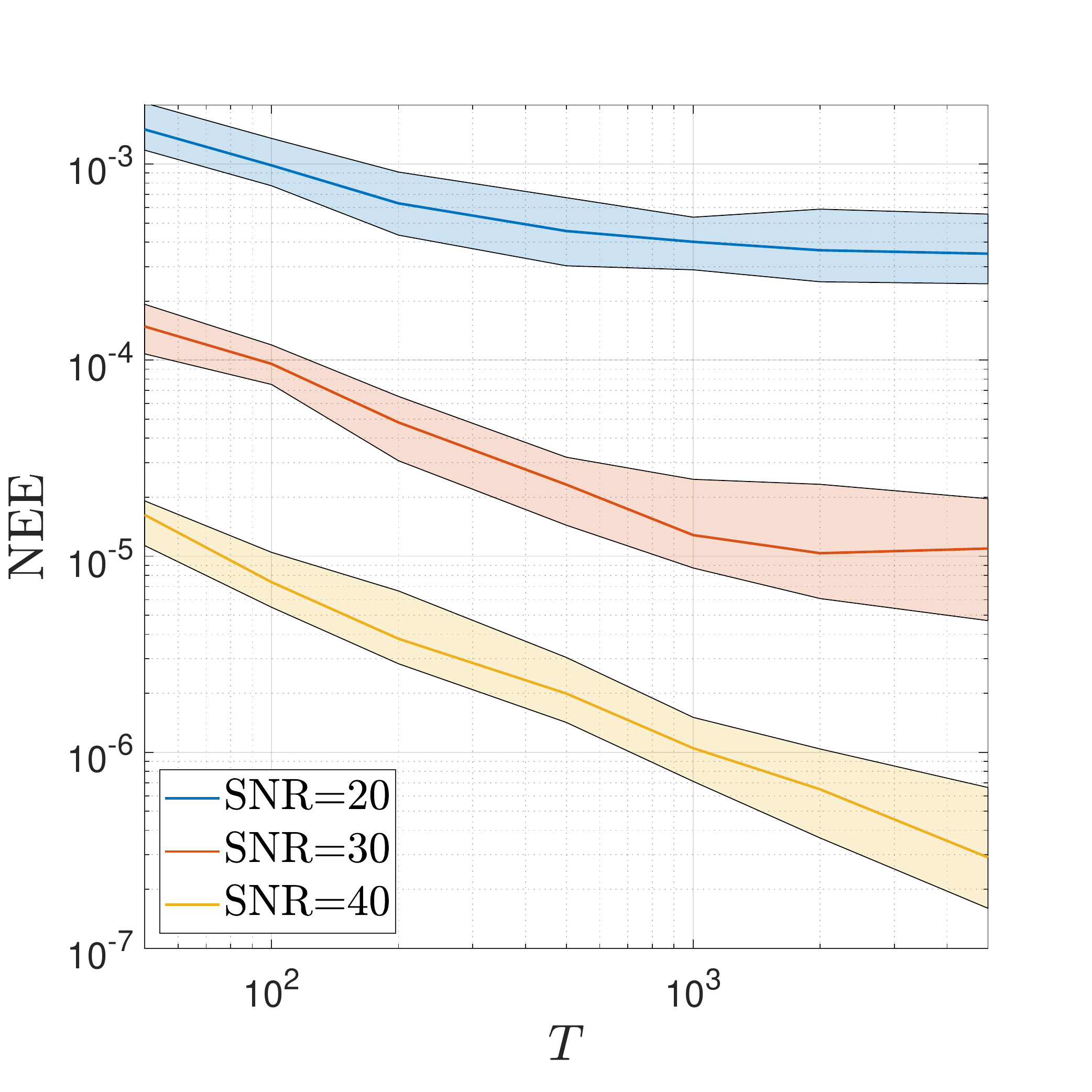}
		\caption{Normalized equation error versus number of training data points for three different SNRs. The central line shows the median NEE for 25 model realisations and the shaded region shows the upper and lower quartiles. The SNR is measured in decibels.\label{fig:lree_consistency}}
	\end{figure}
	
	To explore the effect of noise on the consistency of LREE, we have plotted NEE versus the size of the dataset for varying noise level (measured in decibels) in Figure \ref{fig:lree_consistency}.
	If we had a consistent estimator of the explicit model \eqref{eq:global system},  we would expect to see $\underset{T\rightarrow\infty}{\lim} NEE = 0$ with consistent slope for all SNR levels. What we in fact observe, however, is that in noisier conditions the NEE initially decreases and then plateaus at a certain level. 
	\tb{
		This phenomena can also be seen in \cite[Sec. IV]{Umenburger2018Convex}, where LREE produces models biased towards being too stable, even in the infinite data limit.
	}

	\subsection{Identification of Traffic Networks}
	
	In this section we examine a potential application of our approach, the identification of a traffic network. The dynamics of traffic networks are thought to be monotone when operating in the free flow regime \cite{lovisari2014stability}. Note that monotonicity of some traffic models is lost when certain nodes are congested \cite{como2016convexity}. 
	
	The data are generated using the model in \cite{lovisari2014stability}, which is not in the proposed model set. Hence this section provides a test of robustness of the proposed approach to modelling assumptions.
	
	For this application, we consider using equation error as a surrogate for simulation error. Model performance is therefore measured using Normalized Simulation Error (NSE):
	\begin{equation}
	\text{NSE} = \frac{\sum_{t}|x_t - \tilde{x}_t|^2}{\sum_{t}|\tilde{x}_t|^2},
	\end{equation}
	where $x_t$ are the simulated states. 
	
	We will first introduce the model, then study the effect of the model constraints by comparison to existing methods, and finally examine scalability to large networks.

	\subsubsection{ Simulation of a traffic network \label{sec:traffic simulation}}
	
	
The dynamics are simulated over a graph (e.g. Fig.  \ref{fig:planar_graph}), where, each node $i$ represents a road with state corresponding to the density of traffic on the road, denoted $\rho^i$. Nodes marked \textit{in} allow cars to flow into the network, and nodes marked \textit{out} allow cars to flow out of the network. Each edge $(i,j)$ is randomly assigned a turning preference denoted $R_{ij}$ such that $\sum_{i} R_{ij} = 1$ (this ensures that the total number of cars at each intersection is conserved). Each node $i$ has a capacity of $C_i = 1$. 
Vehicles transfer from roads $i$ to $j$ according to the routing policy,
$$
f_{i\rightarrow j}(\rho)=R_{ji}d_{i}(\rho^{i})\min\left\{1,{s_{j}(\rho_{j})\over\sum_{k\in \upstream{\Nodes}}R_{kj}d_{k}(\rho^{j})}\right\},
$$
where $d_i(\rho) = \min(10, \rho)$ and $s_i(\rho) = \max(2 C_i - \rho, 0)$ are monotone demand and supply curves for road $i$. The dynamics of the complete system are then found to be 
\begin{equation} \label{eq:traffic dynamics}
\dot{\rho}^{i} = f^i_{in} -  f^i_{out},
\end{equation}
where
$$
f^i_{in} = \begin{cases} 
u^i &, ~ i \in \text{in} \\
\sum_{j\in \upstream{\Nodes}} f_{j\rightarrow i} &, ~  i \notin \text{in}
\end{cases}
$$

$$
f^i_{out} = \begin{cases}
d_i(\rho_i) &, ~ i \in \text{out}\\ 
\sum_{j\in \downstream{\Nodes}} f_{i\rightarrow j} &, ~ i \notin \text{out}.
\end{cases}
$$

\begin{figure}
	\centering
	\includegraphics[trim = {2.5cm 9cm 3cm 9cm}, clip, width=0.9\linewidth]{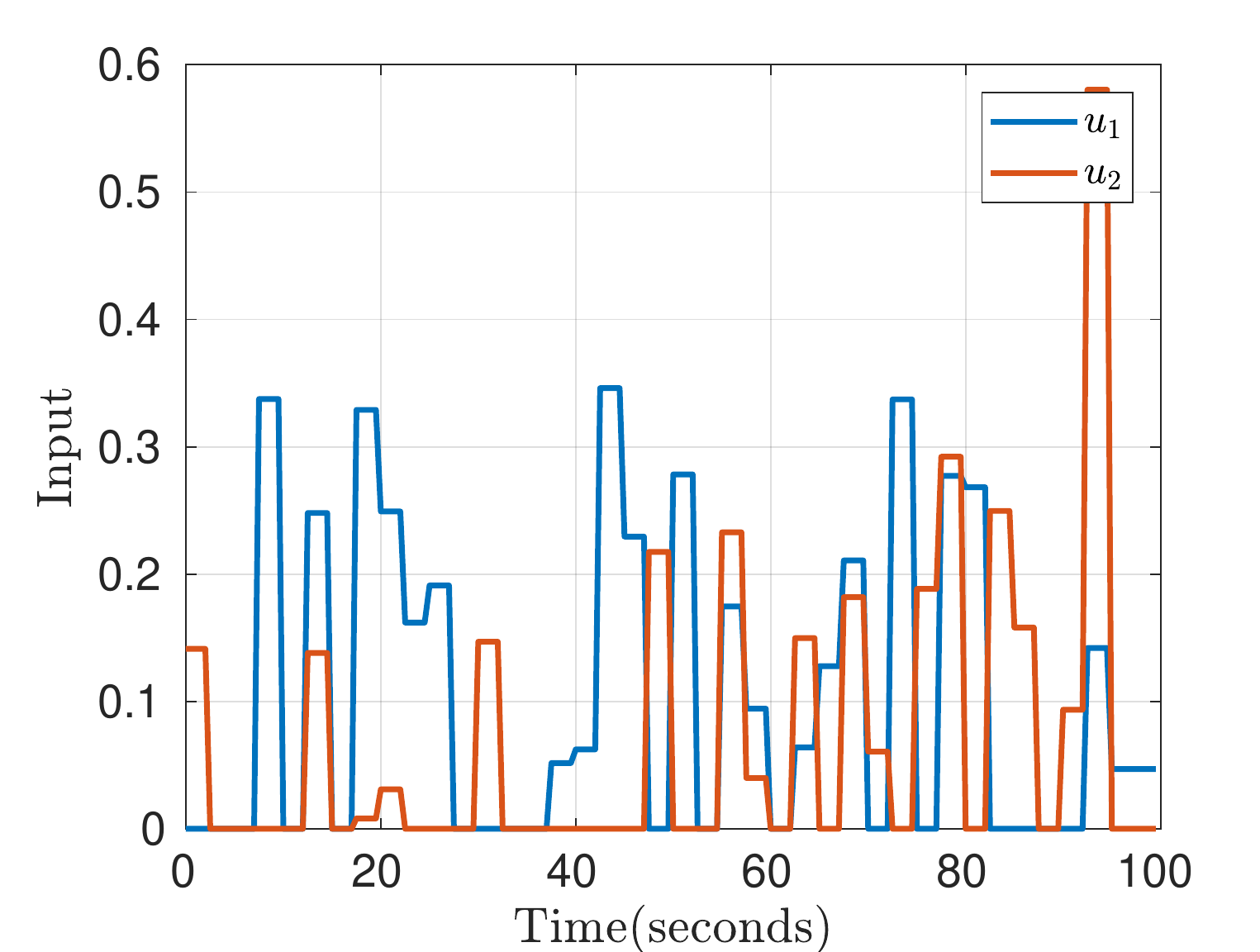}
	\caption{Example input signal to network $(\mu_u = 0, \sigma_u = 0.2)$.}
	\label{fig:network input}
\end{figure}
The input nodes $i \in \textit{in}$ take a time varying input $u^i$.
We use the following method to generate data sets of size $T$:
\begin{enumerate}[(i)]
	\item First, we generate an input signal for each $u^i$ of size $T$. This signal changes value every 5 seconds to a new value that is normally distributed with mean $\mu_u$ and standard deviation $\sigma_u$. Negative values of $u$ are set to zero. An example input signal is shown in Fig.  \ref{fig:network input}.	
	\item The dynamics \eqref{eq:traffic dynamics} are integrated over $t_{f}$ seconds.
	\item  A training set of size $T = 2t_{f}$ is generated by sampling every 0.5 seconds.	
\end{enumerate}

	\subsubsection{ Regularization Effect of Model Constraints\label{sec: results-constraints}}
	In this section we will explore the effects of introducing monotonicity, positivity, and contraction constraints.

	Introducing model constraints limits the expressivity of our model. Consequently, one might expect the estimator bias to increase and the variance to decrease \cite[Chapter 7]{friedman2001elements}.  Empirical evidence in this section suggests that a  judicious choice of constraints can reduce the variance with a minimal increase in bias.

	\begin{figure}
		\centering
		\includegraphics[trim = {2cm 16.5cm 2cm 3cm}, clip,width=0.8\linewidth]{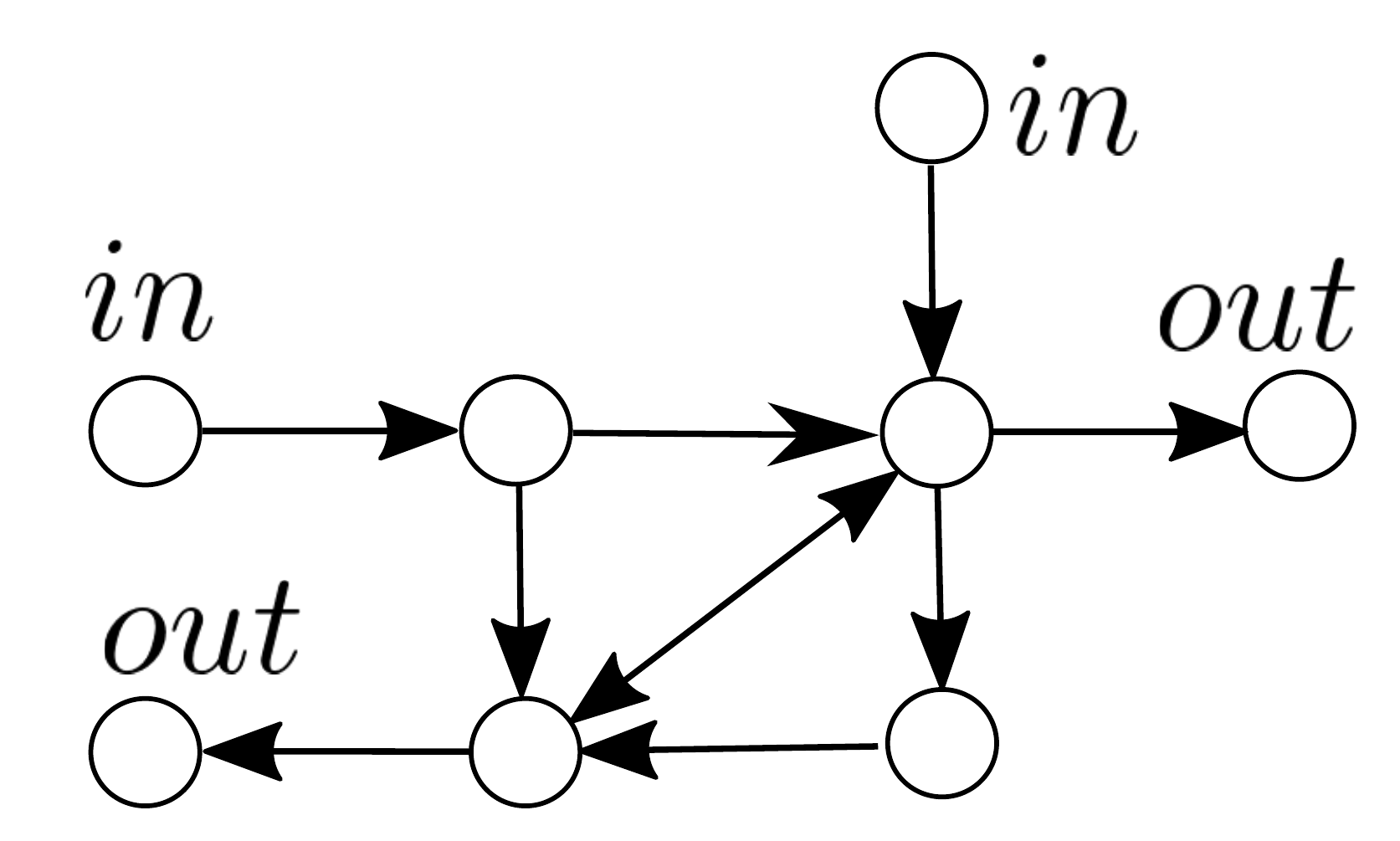}
		\caption{\label{fig:planar_graph}A small traffic network. Each node represents a road and each link represents an intersection. }
	\end{figure}
	
	Using the method outlined in Section \ref{sec:traffic simulation} for simulating a traffic network and the graph depicted in Fig.  \ref{fig:planar_graph},  we generate 100 different training sets of size $T=1000$ with $\mu_u = 0, \sigma_u = 0.2$ and then compare the results on three different validation sets. 
	The first validation set has inputs generated with parameters $\mu_u = 0, \sigma_u = 0.2$ (the same as the training set). The second and third validation sets have parameters $\mu_u = 0, \sigma_u = 0.3$ and $\mu_u = 0, \sigma_u = 0.4$ respectively. These are used to test the generalizability of our model to inputs outside the training set.
	
		\begin{figure}
		\centering
		\subfloat[ \label{fig:model constraint training NSE} Training set ($\sigma_u = 0.2, \mu_u = 0$) over 100 realizations.]
		{
			\includegraphics[trim = {0.5cm, 0cm, 1cm, 0cm}, clip, width =\linewidth]{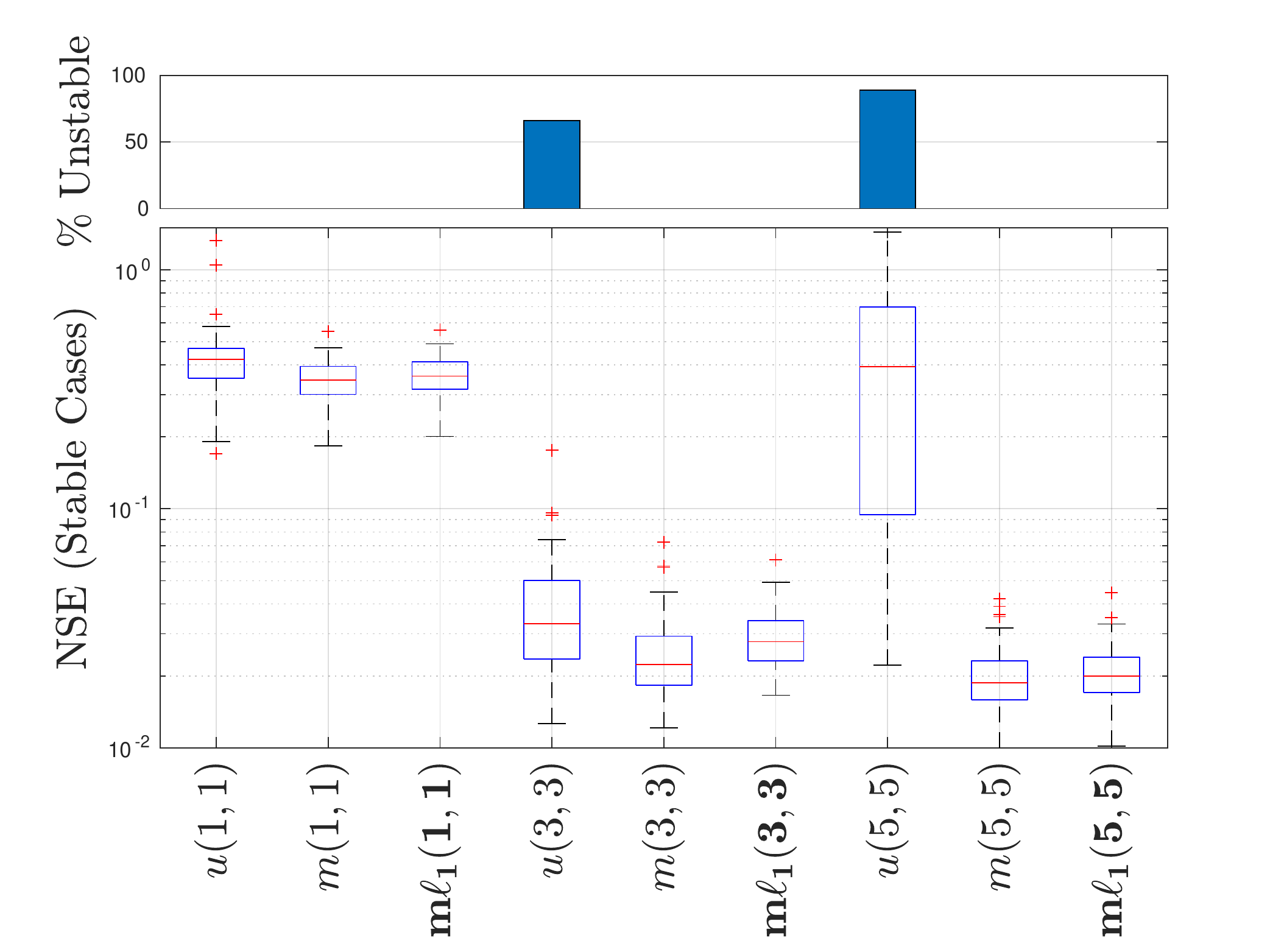}
		}
		
		\subfloat[ \label{fig:model constraint validation1 NSE} Validation set 1 ($\sigma_u = 0.2, \mu_u = 0$) over 100 realizations.]
		{
			\includegraphics[trim = {0.5cm, 0cm, 1cm, 0cm}, clip, width =\linewidth]{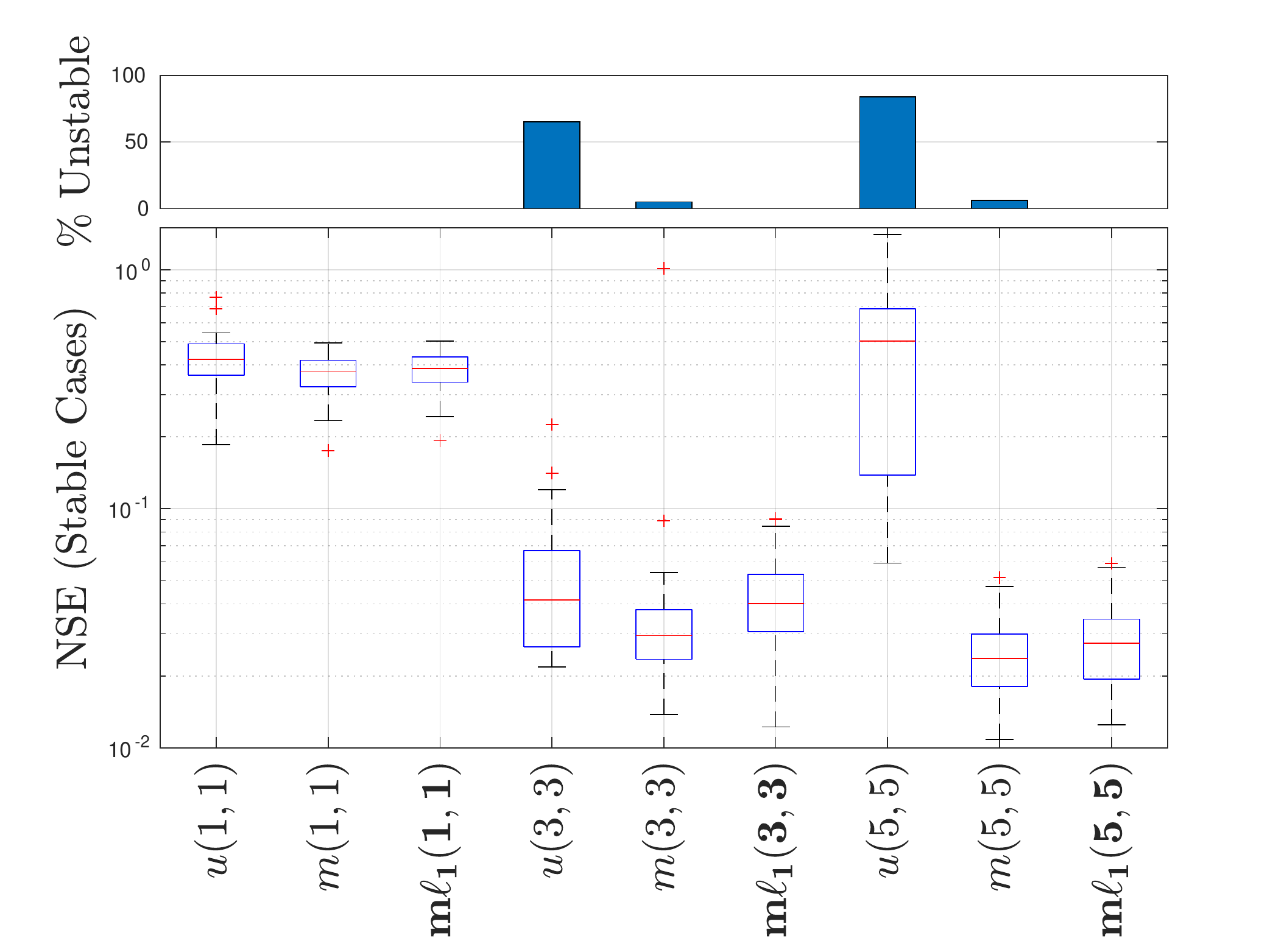}
		}
		
		\subfloat[\label{fig:model constraint validation2 NSE} Validation set 2 ($\sigma_u = 0.3, \mu_u = 0$) over 100 realizations.  ]
		{
			\includegraphics[trim = {0.5cm, 0.0cm, 1cm, 0.0cm}, clip, width =\linewidth]{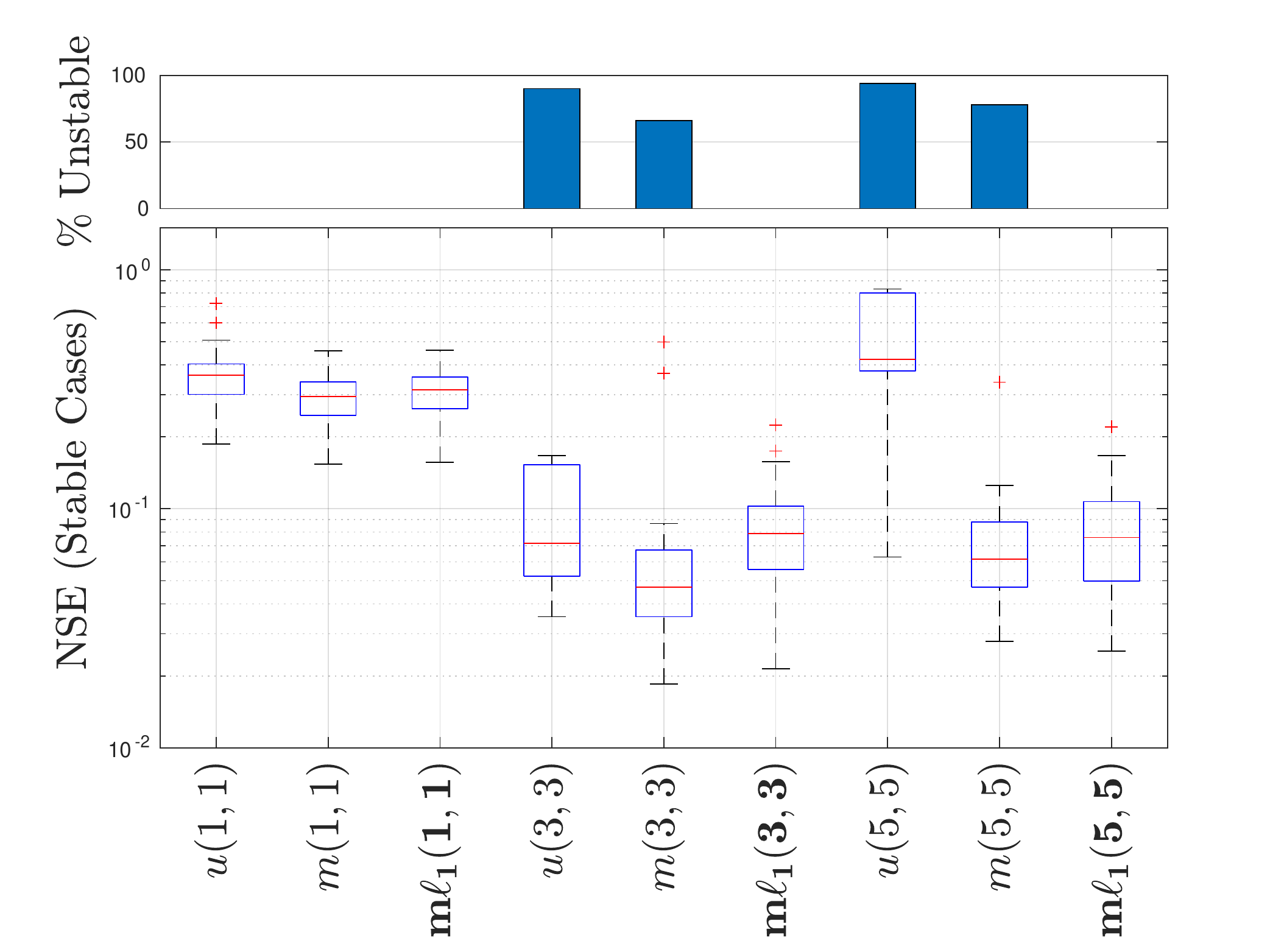}
		}
		\caption{ \label{fig:model_constraints} Box plots showing normalized simulation error for 100 model realizations for different behavioural constraints. The bar graph shows the percentage of models that displayed instability.}
	\end{figure}
	\begin{figure}
		\centering
		\subfloat[ \label{fig:model_comparison_training} Training set ($\sigma_u = 0.2, \mu_u = 0$) over 100 realizations.]
		{
			\includegraphics[trim = {0.5cm, 0cm, 1cm, 0cm}, clip, width =\linewidth]{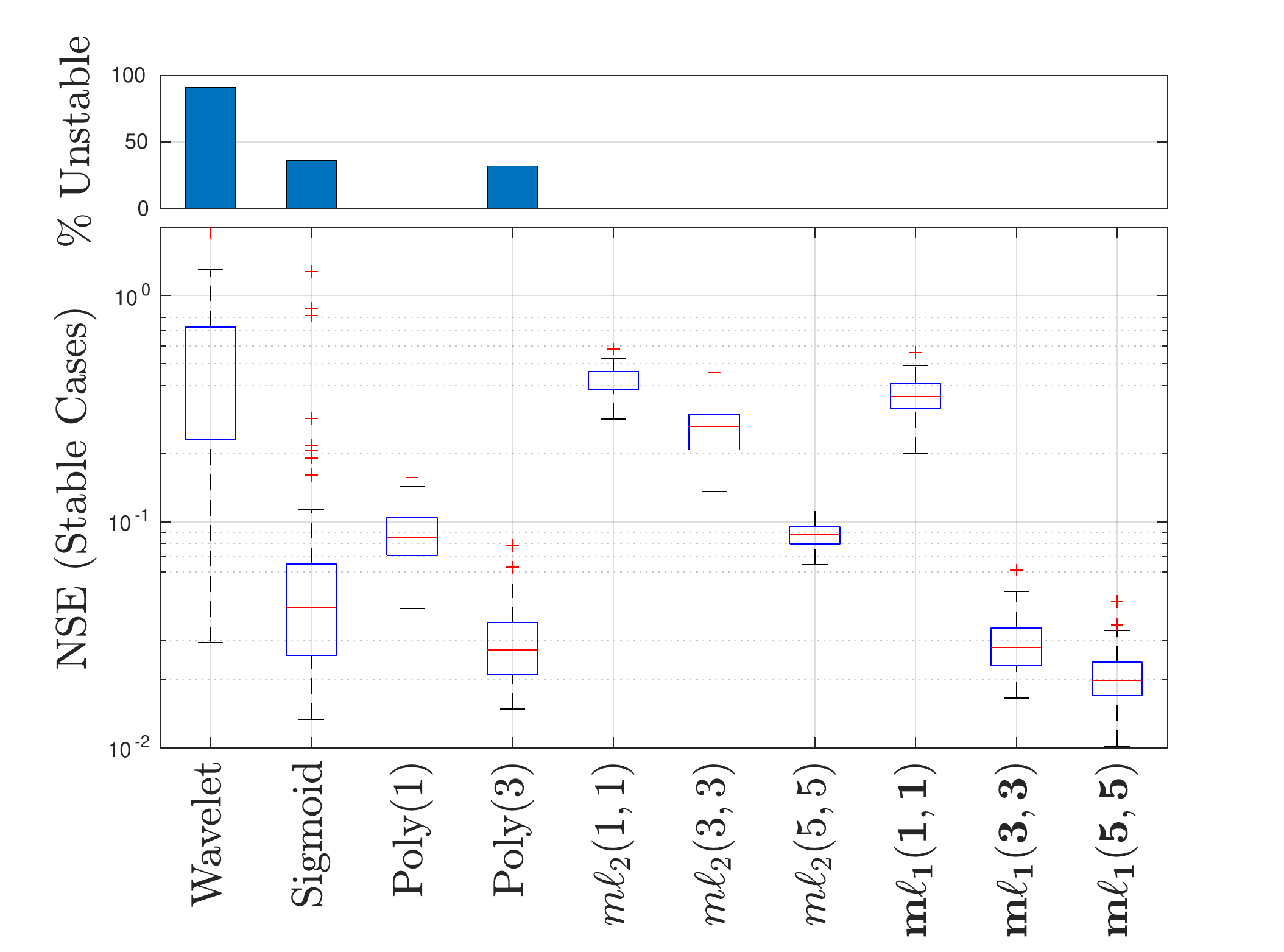}
		}
		
		\subfloat[ \label{fig:model_comparison_val1} Validation set 1 ($\sigma_u = 0.2, \mu_u = 0$) over 100 realizations.]
		{
			\includegraphics[trim = {0.5cm, 0cm, 1cm, 0cm}, clip, width =\linewidth]{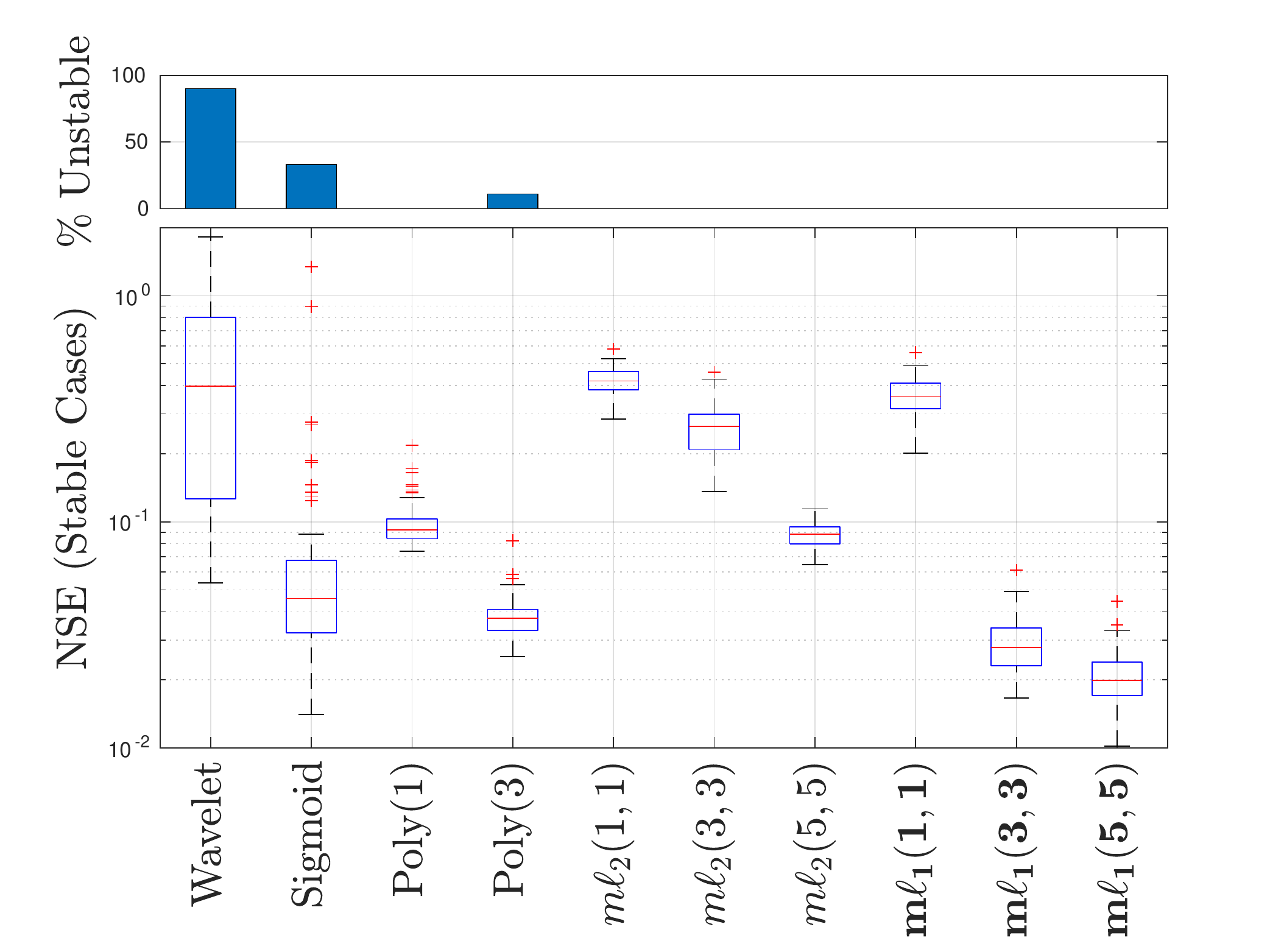}
		}
		
		\subfloat[\label{fig:model_comparison_val2} Validation set 2 ($\sigma_u = 0.3, \mu_u = 0$) over 100 realizations.  ]
		{
			\includegraphics[trim = {0.5cm, 0cm, 1cm, 0cm}, clip, width =\linewidth]{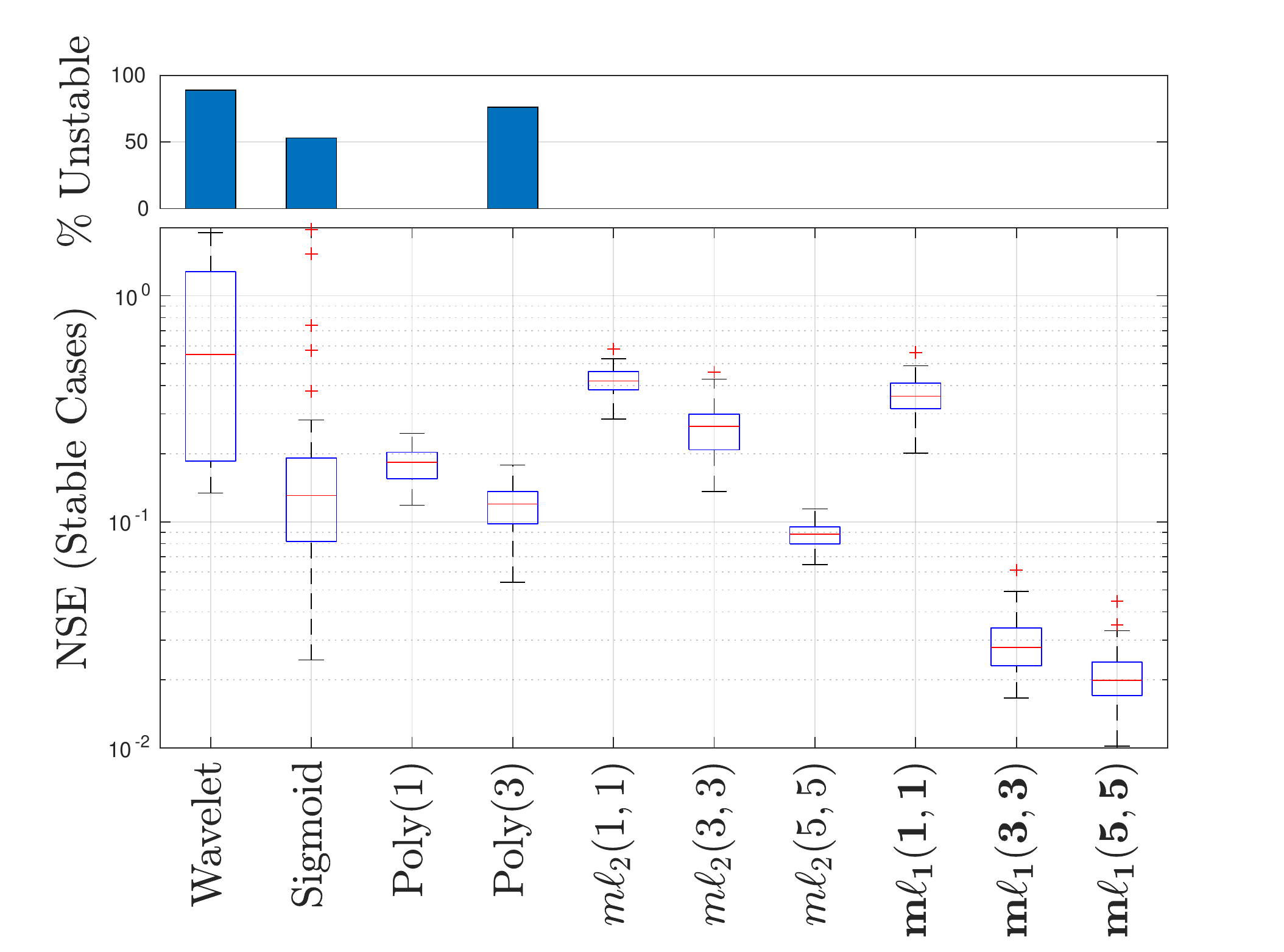}
		}
		\caption{ \label{fig:model_comparison} Box plots showing normalized simulation error for 100 model realizations for different model structures. The bar graph shows the percentage of models that displayed instability.}
	\end{figure}

\tb{	In figures \ref{fig:model_constraints} and \ref{fig:model_comparison}, we have plotted the $NSE$ on both the training set and validation sets 1 and 2 for our proposed model sets, the polynomial model, the NARX models and the model set $m\ell_2$.
	The percentage of total models that displayed instability is indicated in both the bar graph in the upper portion of the figures. }

	
	In all cases, the identified linear models performed poorly. This is unsurprising as the true system is highly non-linear. 
	
\tb{	Comparing the models $m\ell_1$ and $m\ell_2$ with the remaining models, we can see that the stability constraints have a regularizing effect where increasing the degree of the polynomials reduces the median NSE; in other words, increasing model complexity improves model fidelity. 
	The other models on the other hand perform worse with increasing the complexity. This is most clearly seen in the models $u$, where increase the polynomial degree results in poorer fits on validation data.}

	
	Our results also suggest that model stability constraints significantly improve robustness.
	Without stability constraints, a model that appears stable during training may turn out to be unstable under a slight shift in the input data distribution. This can be seen most clearly in the models $m(5,5)$ and $Poly(3)$, where on the training data distribution, most models are stable.
	However, increasing the variance of the inputs to the network results a large number of unstable models with unbounded NSE, c.f. Fig. \ref{fig:model constraint validation2 NSE} and Fig. \ref{fig:model_comparison_val2}. 
	Further evidence is shown in Table \ref{table: Unstable Models }, where we can see that once the variance of the input data doubles, almost all models that do not have stability constraints are unstable.
	
	\tb{To compare to a standard approach, we also compare to wavelet and sigmoid NARX models fit using the Matlab system identification tool box. The resulting $NSE$ is shown in Fig. \ref{fig:model_comparison} and show the number of models producing unstable models and negative state estimates in tables \ref{table: Unstable Models } and \ref{table: Negative states} respectively. While we observed extremely high performance of the individually identified sub-systems, simulating the network interconnection of those sub-systems produces
	many unstable models, many negative state estimates and poor quality of fit.}

	\begin{table*}[t]
		\centering
		\begin{tabular}{|c|c|c|c|c|c|c|c|c|c|c|c|c|c|c|}
			\hline      &   $\bm{m \ell_1(1,1)}$& $m(1,1)$ & $u(1,1)$ &  $\bm{m \ell_1(3,3)}$& $m(3,3)$ & $u(3,3)$ & $\bm{m \ell_1(5,5)}$& $m(5,5)$ & $u(5,5)$ & Wavelet & Sigmoid\\ \hline
			train. ($\sigma_u = 0.2$) &0\%   & 0\% & 2\%  & 0\%  & 0\%  & 66\% & 0\%  & 0\%  & 89\% & 88\% & 36\% \\ \hline
			val. 1 ($\sigma_u = 0.2$)   & 0\%  & 0\%  & 0\% & 0\%  & 5\%  & 65\% & 0\%  & 6\% &  84\% & 87\% & 31\% \\\hline
			val. 2 ($\sigma_u = 0.3$)   & 0\%  & 0\%  & 0\% & 0\%  & 64\% & 90\% & 0\%  & 78\% & 94\% & 88\% & 51\% \\ \hline
			val. 3 ($\sigma_u = 0.4$)   & 0\%  & 0\%  & 0\% & 0\%  & 88\% & 95\% & 0\%  & 97\% & 89\% & 88\% & 63\% \\ \hline
			
		\end{tabular}
		\caption{\label{table: Unstable Models } Percentage of unstable models that diverged on training and validation data. In each case the input $u$ has $\mu_u = 0$.}
	\end{table*}
	
	\begin{table*}[t]
		\centering
		\begin{tabular}{|c|c|c|c|c|c|c|c|c|c|c|c|}
			\hline          & $\bm{m \ell_1(1,1)}$& $m(1,1)$ & $u(1,1)$& $\bm{m \ell_1(3,3)}$& $m(3,3)$ & $u(3,3)$ & $\bm{m \ell_1(5,5)}$& $m(5,5)$ & $u(5,5)$ & Wavelet & Sigmoid \\ \hline
			train. ($\sigma_u = 0.2$) & 0\%  & 0\%  & 2\% & 0\%  & 0\%  & 66\%  & 0\%  & 0\%  & 89\% & 100\% & 91\% \\ \hline
			val. 1 ($\sigma_u = 0.2$)& 0\%  & 0\%  & 0\% & 0\%  & 0\%  & 65\%  & 0\%  & 0\%  & 84\% & 100 \% & 93\% \\ \hline
			val. 2 ($\sigma_u = 0.3$)& 0\%  & 0\%  & 2\% & 0\%  & 0\%  & 90\%  & 0\%  & 0\%  & 94\%& 100 \% & 99\% \\ \hline
			val. 3 ($\sigma_u = 0.4$)& 0\%  & 0\%  & 1\% & 0\%  & 0\%  & 95\% & 0\%  & 0\%  & 99\%& 100 \% & 100\% \\ \hline
		\end{tabular}
		\caption{\label{table: Negative states} Percentage of total models that predicted negative states. In each case the input $u$ has $\mu_u = 0$. }
	\end{table*}

	For positive linear systems, both $\mc{}$ and $\Theta_{m \ell_2}$ are parameterizations of the same set of models. This is not the case for nonlinear monotone systems and the choice of parametrization impacts the resulting model performance. This can be seen in Fig.  \ref{fig:model_comparison_training}, Fig. \ref{fig:model_comparison_val1} and Fig. \ref{fig:model_comparison_val2} where the models fit using our proposed $\ell_1$ contraction constraint outperform those fit using the previously-proposed $\ell_2$ contraction constraint.
	
	Finally, looking at Table \ref{table: Negative states}, we can see that when models were not constrained to be positive $u$ and $Poly$, a large number of models producing negative state estimates were identified. This can lead to non-sensical results in many applications, and prevents the application of synthesis methods that depend on monotonicity.
	
	\subsubsection{Scalability Comparison of $\ell_1$ and $\ell_2$ contraction\label{sec:l1 versus l2 traffic}}
	We now explore the scalability of the $\ell_1$ and $\ell_2$ contraction constraints for nonlinear models.
	
	We construct traffic networks consisting of $N = P + 2M$ nodes by placing $P$ points randomly in a unit square and triangulating. $M$ in nodes and $M$ out nodes are then randomly assigned throughout the network.
	We generate training data using the method described in Section \ref{sec:results} with $T = 600, \mu_u = 0, \sigma_u = 0.4$ and a corresponding validation set.
	We then fit models $m \ell_1(3,3)$ and $m \ell_2(3,3)$ using an interior point method. This is repeated 5 times for a varying number of nodes.
	

	Figure \ref{fig:runtimes} shows a plot of the time taken to solve each problem versus the total number of nodes. 
	We observe that fitting models with an $ \ell_2$ contraction constraint has a complexity $\mathcal{O}[N^3]$ in the number of nodes while models using the $ \ell_1$ contraction constraint have a complexity of $\mathcal{O}[N^{1.5}]$ in the number of nodes.
	The improved complexity of the $ \ell_1$ constraint is a result of its separable structure. 
	
	The validation NSE versus the number of agents is shown in Fig. \ref{fig:NSE_vs_agents} for the model set in $m\ell_1(3,3)$. We observe no deterioration of model performance as the number of agents increases, suggesting that our method can be effective when scaled to large networks.
	

	\begin{figure}[]
		\centering
		\includegraphics[trim = {0cm, 0cm, 0.75cm, 0cm}, clip,width = 0.9\linewidth]{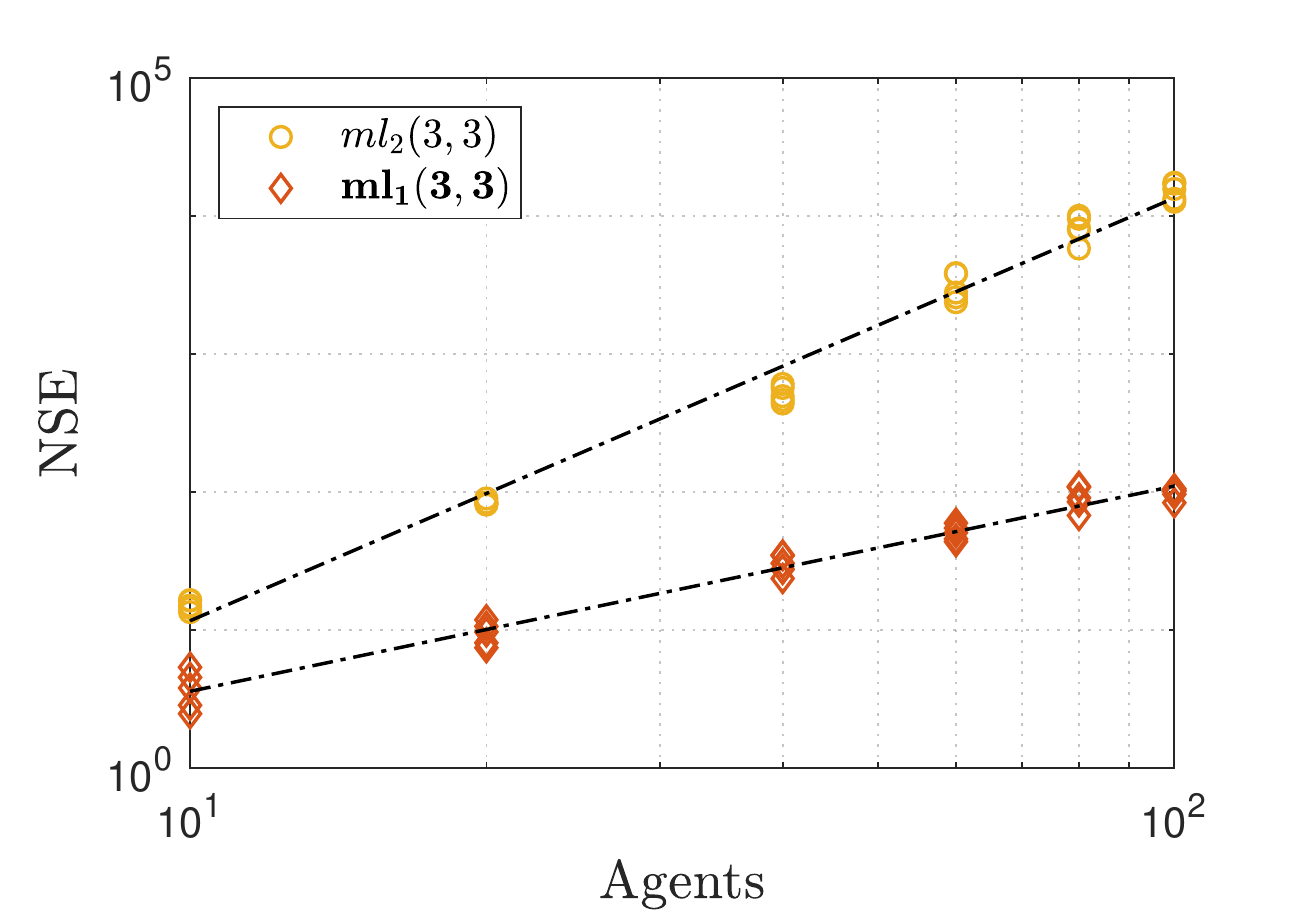}
		\caption{\label{fig:runtimes} Computation time for models $m\ell_2(3,3)$ and $m\ell_1(3,3)$ for a varying system size. The slopes of the lines are $3.06$ and $1.49$ respectively.}
	\end{figure}
	
		\begin{figure}[]
		\centering
		\includegraphics[trim = {0cm, 0cm, 0.75cm, 0cm}, clip,width = 0.9\linewidth]{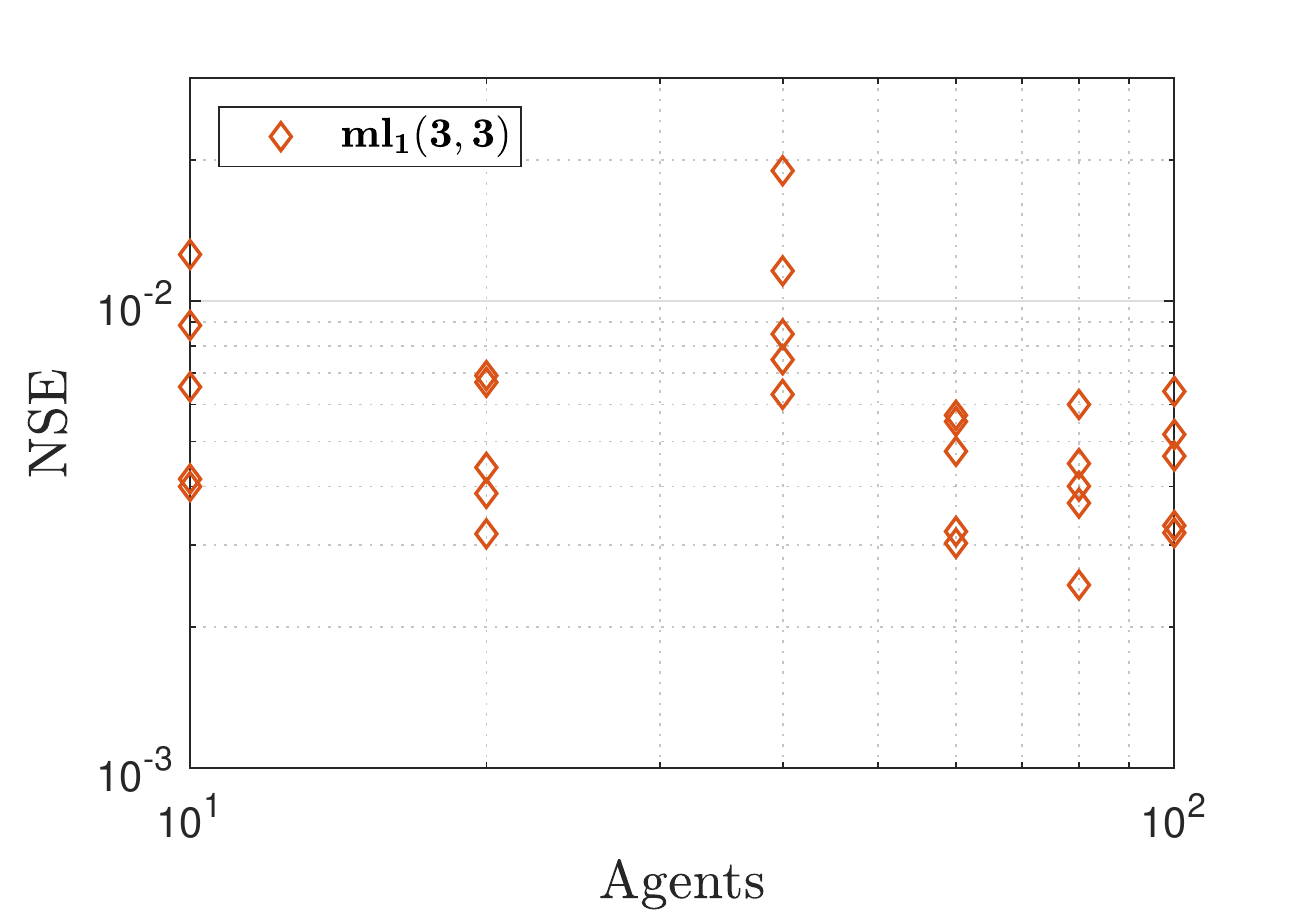}
		\caption{\label{fig:NSE_vs_agents} NSE for models $m\ell_1(3,3)$ for varying system size.}
	\end{figure}
	

	\subsubsection{Scalability Compared to Interior Point Methods}
	
	We conclude our numerical experiments with a comparison of the computational complexity of the proposed distributed algorithm  to centralized optimization via standard interior point methods. 
	
	We introduce additional notation to distinguish between the centralized and distributed algorithms. We will use a subscript $C$ to refer to models fit using the off-the-shelf interior point method. The subscript $D$ is used to denote models fit using ADMM. For example, $m\ell_1(3,3)_D$ is the problem of fitting the model $m\ell_1(3,3)$ solved using the distributed algorithm.
	

	
	To control for the number of neighbors of each node, we generate random, connected, regular graphs of size $N$ and degree $4$ and randomly assign $\frac{P}{2}$ \textit{in} nodes and $\frac{P}{2}$ \textit{out} nodes. Training data is generated according to Section \ref{sec:traffic simulation} with $T=500$ and $\sigma_u = 0.2$. 
	
	We then solve the problems $m\ell_1(3,3)_C$ and $m\ell_1(3,3)_D$ using the stopping criteria from 	\cite[Section 3.3]{boyd2011distributed} ($\epsilon_{abs} = 10^{-4}$, $\epsilon_{rel} = 10^{-3}$).
	
	The results are displayed in Fig. \ref{fig:runtime_threads_sim}. The line $m\ell_1(3,3)_{D- serial}$ indicates the total time taken to fit a model using ADMM, where the sub-problems \eqref{ADMM iterates 2.1}, \eqref{ADMM iterates 2.2} are solved without parallelization (consecutively, on a single computer). 
	Additionally, we calculate the total time that would be taken if the computation had been distributed among $N$ nodes, indicated by the line $m\ell_1(3,3)_{D - parallel}$.  
	
	While the program $m\ell_1(3,3)_{D- serial}$ takes longer on the selected problems than $m\ell_1(3,3)_{C}$, it has superior scalability with $\mathcal{O}[N^{1.05}]$ compared to $\mathcal{O}[N^{1.36}]$, suggesting that for a larger number of nodes, it will be faster.
	
	Of more interest is $m\ell_1(3,3)_{D - parallel}$ with an observed complexity of $\mathcal{O}[N^{0.05}]$ in the number of nodes. This suggests that if the computation is distributed, the problem can be solved in near constant time. 
	It is important to note, however, that this does not take into account many of the complexities of distributed computing, for example the overhead associated with communication between nodes.

	\begin{figure}[]
		\centering
		\includegraphics[trim = {0cm 0cm 0.5cm 0.5cm}, clip,width = \linewidth]{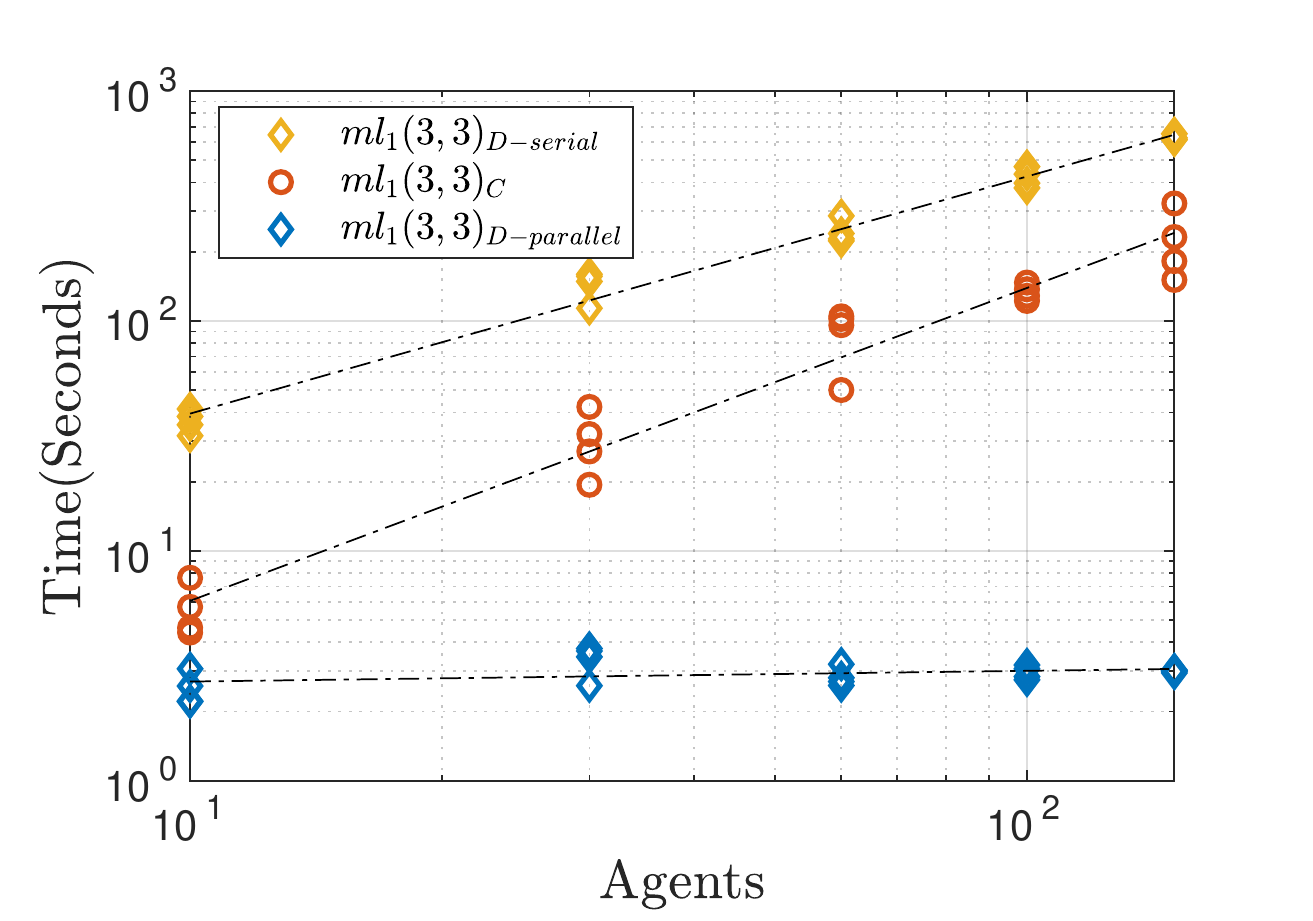}
		\caption{\label{fig:runtime_threads_sim} Runtime of ADMM compared to IPM where the number of threads is one or equal to the number of nodes. When calculating the results for "simulated" distributed computing, ADMM is run in series and time per iteration is taken to be the sum of the maximum times to solve each step. The slopes of the lines are $1.05$, $1.36$ and $0.047$ respectively.}
	\end{figure}

	\section{Conclusion}
	In this paper we have proposed a model set for system identification that allows model behavioural guarantees such as stability (contraction), monotonicity, and positivity. Furthermore, we have introduced a particular separable structure that allows distributed identification and scalability to large networked systems via local node-to-node communication.
	
	We have examined the proposed approach via a selection of numerical case studies including a nonlinear traffic network. The main conclusions are that the approach scales much better than previous approaches guaranteeing stability, and that behavioural constraints such as stability and monotonicity can have a regularising effect that leads to superior model predictions.

	\appendix 
	
	\subsection{Theorem 1} \label{appendix:thm1}
	We use the following lemma  in the proof of Theorem \ref{thm: L1 contracting models}:
	
	\begin{lemma} \label{cor:ell p stability}
		Suppose that for the system \eqref{eq:implicit model}, there exists a weighted differential $ \ell_1$ storage function $V_t = |E(x_t, u_t) \delta_t|_1$, where $E:\mathbb{R}^{n}\times \mathbb{R}^m \rightarrow \mathbb{M}^{n}$  such that $V_{t+1}\leq  \alpha V_t$ and there exists some $K\succ0$  such that $|\delta_t|_1 \prec K |E_t\delta_t|_1$,
		then the system is contracting in the sense of definition \ref{def:Incremental Lp Stability}. 
	\end{lemma}	
	\begin{proof}
		Consider the family of solutions to \eqref{eq:implicit model}, parametrized by $\rho \in [0,1]$, having initial conditions $\rho x_1(0) + (1-\rho)x_2(0)$ and input $u(t)$, denoted $x_\rho(t)$.
		
		Define $\delta_\rho(t) = \parDiff{x_\rho(t)}{\rho}$. Now, consider:
		
		\begin{align*}
		|x_1(t) - x_2(t)|_1 &= \left|\int_{0}^{1} \delta_\rho(t) d\rho\right|_1 \\
		&\leq \int_0^1 |\delta_\rho(t)|_1 d\rho \\
		&\leq \int_0^1 K |E_{t} \delta_\rho(t)|_1 d\rho
		\end{align*}
		By assumption, $V_{t+1} \leq \alpha V_t$ which means that $|E_{t}\delta(t)|_1 \leq \alpha|E_{t-1}\delta_{t-1}|_1$. This inequality can be applied repeatedly to give: 
		\begin{equation*}
		|x_1(t) - x_2(t)|_1  \leq  K\alpha^t \int_0^1|E_0\delta_\rho(0)|_1 d \rho
		\end{equation*}
		Taking $b(x_1(0), x_2(0)) = K\int_0^1 |E_0\delta_\rho(0)|_1 d \rho$ gives Definition \ref{def:Incremental Lp Stability}.
	\end{proof}

	\begin{proof}[Proof of Theorem \ref{thm: L1 contracting models}]		
		First we will show well-posedness and monotonicity. We will then prove stability of monotone contracting systems and finally just contracting systems. For brevity of the equations, we will use a subscript $t$ to refer to the evaluation of a function at a specific time, so $E_t = E(x_t, u_t)$.

		\textit{Well-posedness:}  Assume \eqref{eq:well-posed}. Since $E$ is a non-singular M matrix, there exists a diagonal matrix $D$ such that $ED + DE^\top \succ 0$. Well posedness follows from the same argument as \cite[Theorem 5]{tobenkin2017convex}.
		
		\textit{Monotonicity:} Assume \eqref{eq:monotonicity}. Since $E$ as an M-matrix, it is inverse positive and $E^{-1}F \geq 0$. The differential dynamics of the explicit system \eqref{eq:explicit system} can be written as $\delta_{x_{t+1}} = E_{t+1}^{-1} F \delta_{x_{t}}$. Therefore, the explicit system is monotone.

		\textit{Contraction:}	Assume conditions \eqref{eq:slack_condition} and \eqref{eq:l1_contraction_polytone}. Condition \eqref{eq:slack_condition} implies that 
		\begin{equation}
		|F(x,u)| \leq S(x,u).
		\end{equation}
		Condition \eqref{eq:l1_contraction_polytone} then implies,
		\begin{gather}
		\bm{1}^\top (\alpha E(x,u) -  S(x,u)) \geq 0, \\
		\implies \bm{1}^\top (\alpha E(x,u) -  |F(x,u)|) \geq 0, \\
		\implies \bm{1}^\top (\alpha -  |F(x,u)|E^{-1}(x,u)) \geq 0, \\
		\implies \bm{1}^\top (\alpha -  |F(x,u)E^{-1}(x,u)|) \geq 0, \\
		\implies(\alpha -  ||F(x,u)E^{-1}(x,u)||_1) \geq 0,
		\end{gather}
		where $||\cdot||_1$ is the induced matrix norm , $||M|| := \max_j \sum_{i}{M^{ij}}$.
		Stability follows from the same argument as in the proof of Theorem \ref{thm: L1 contracting models}. 
		Multiply by $|E_t \delta_{t}|_1$, we get:
		\begin{gather}
		\left( \alpha -  ||F(x,u)E^{-1}(x,u)||_1) \right)|E_t \delta_{t}|_1  \geq 0, \\
		\implies \alpha|E_t \delta_{t}|_1 -  |F(x,u)E^{-1}_t E_t \delta_{t}|_1) \geq 0, \\
		\implies |F_t \delta_t|_1 - \alpha|E_t \delta_{t}|_1 \leq 0, \\
		\implies |E_{t+1} \delta_{t+1}|_1 - \alpha|E_t \delta_{t}|_1 \leq 0 .
		\end{gather}
		Contraction then follows from Lemma 1 with contraction metric $V_t = |E(x_t, u_t) \delta_t|_1$.
		\textit{Monotonicity and Contraction}
		Finally, to see how contraction follows from \eqref{eq:monotonicity} and \eqref{eq:l1 contraction condition}, note that they imply conditions \eqref{eq:slack_condition} and \eqref{eq:l1_contraction_polytone}.
		
%
%
	\end{proof}

	\subsection{Proof of Theorem \ref{thm:admm_is_separable}} \label{appendix:admm_is_separable}
	\begin{proof}
		The first step \eqref{ADMM iterates 1.1} can be broken up into the following sum:
		$$
		\theta(k+1) = \arg \min_{\theta} \sum_{i=1}^N \hat{J}^i_{ee}(\upstream{\theta}) + \frac{\rho}{2}||\upstream{\theta} - \upstream{\phi}(k)  + \upstream{u}(k)||^2,
		$$
		which is equivalent to the $N$ optimization problems in \eqref{ADMM iterates 2.1}. The second step \eqref{ADMM iterates 1.2} can be written as 
		\begin{equation}		
		\phi(k+1) = \arg \min_\phi \mathcal{I}_{\mc{}}(\phi) + \sum_{i=1}^N  \frac{\rho}{2}||\downstream{\theta}({k+1}) - \downstream{\phi}  + \downstream{u}(k)||^2.
		\end{equation}
		We will show that the indicator function can be written as a sum over $i = 1,...,N$ indicator functions each depending on $\downstream{\phi}$.	
		Splitting it up in terms of the individual constraints, we get
		\begin{multline} \label{eq:theta_l1_constraint}
		\mathcal{I}_{\mc{}}(\phi) = \mathcal{I}_{F_x \geq 0 }(\phi) + \mathcal{I}_{F_u \geq 0 }(\phi) + \mathcal{I}_{E \in \mathbb{M} }(\phi) + \\\mathcal{I}_{ \bm{1}^\top ( \alpha E - F \geq 0) }(\phi).
		\end{multline}
		The first two terms can be written as element-wise SOS constraints.
		The last two terms can then be written as a sum over the columns of the matrices $E$ and $F$. We can therefore right \eqref{eq:theta_l1_constraint} as:
		$$
		\mathcal{I}_{\mc{}}(\phi) = \sum_i \mathcal{I}_{\downstream{\phi} \in \Theta^i_{m \ell_1}}(\downstream{\phi})
		$$
		where,
		\begin{multline*}
		\mathcal{I}_{\downstream{\phi} \in \Theta^i_{m \ell_1}}(\downstream{\phi}) =
		\mathcal{I}_{ \alpha  E^{ii}  - \sum_{k\in \downstream{\Nodes}} F^{ki} \geq 0 }(\downstream{\phi}) + \\ 
		\mathcal{I}_{ E^{ii} + {E^{ii}}^\top  > \epsilon  }(\downstream{\phi}) + \\
		\sum_{k\in \downstream{\Nodes}} \mathcal{I}_{E^{ki} \geq 0}(\downstream{\phi}) + 	\sum_{k\in \downstream{\Nodes}} \mathcal{I}_{F^{ki} \geq 0}(\downstream{\phi}).
		\end{multline*}
	\end{proof}

	\subsection{Theorem \ref{thm:discrete_time_positive_linear}} \label{appendix:thm2}
	Sufficiency follows from Theorem 1.
	
	We now prove necessity, i.e. that if a positive linear system is Schur stable, then $\theta \in \Theta_{m \ell_1}$.
	Suppose a matrix $A$ is Schur stable. Then by \cite[proposition 2]{rantzer2011distributed},
	there exists some $z> 0$ such that $z^\top A - z^\top  < 0 $. We can always rescale $z$ such that $z^\top A - z^\top  \leq -\epsilon \bm{1}$. With this $z$, we choose $E = diag(z) \geq 0$ and $F = EA \geq 0$. Then 
	$$
	z^\top A - z^\top  \leq -\epsilon \bm{1}^\top  \implies \bm{1}^\top (F - E) \leq -\epsilon \bm{1}^\top  \implies \theta \in \Theta_{m \ell_1}
	$$

	%
	
	\bibliographystyle{ieeetr}
	\bibliography{Refs}

\end{document}